\documentclass[journal,draftcls,onecolumn,12pt,twoside]{IEEEtranTCOM}
\normalsize

\usepackage{cite}

\ifCLASSINFOpdf
   \usepackage[pdftex]{graphicx}
\else
\fi

\usepackage[cmex10]{amsmath}
\usepackage{array}
\usepackage[caption=false,font=footnotesize]{subfig}
\usepackage{multirow,makecell}
\usepackage{amsthm}
\usepackage{amssymb}
\usepackage{color}

\begin{document}
%
\title{Age of Information in Relay-Assisted \\Status Updating Systems}
%
%
%

\author{Jixiang~Zhang,
        and~Yinfei~Xu,~\IEEEmembership{Member,~IEEE}
\thanks{The authors are with the School of Information Science and Engineering,
Southeast University, Nanjing 210096, China (e-mail: 230179077@seu.edu.cn; yinfeixu@seu.edu.cn).}}

%
%

\markboth{IEEE Transactions on Communications}%
{Submitted paper}
%



\maketitle

\begin{abstract}
In this paper we consider the age of information (AoI) of a status updating system with a relay, where the updates are delivered to destination either from the direct line or the two-hop link via the relay. An updating packet generated at source is sent to receiver and the relay simultaneously. When the direct packet transmission fails, the relay replaces the source and retransmits the packet until it is eventually obtained at the receiver side. Assume that the propagation delay on each link is one time slot, we determine the stationary distribution of the AoI for three cases: (a) relay has no buffer and the packet delivery from relay cannot be preempted by fresher updates from source; (b) relay has no buffer but the packet substitution is allowable; (c) relay has size 1 buffer and the packet in buffer is refreshed when a newer packet is obtained. The idea is invoking a multiple-dimensional state vector which contains the AoI as a part and constituting the multiple-dimensional AoI stochastic process. We find the steady state of each multiple-dimensional AoI process by solving the system of stationary equations. Once the steady-state distribution of larger-dimensional AoI process is known, the stationary AoI distribution is also obtained as it is one of the marginal distributions of that process's steady-state distribution. For all the situations, we derive the explicit expression of AoI distribution, and calculate the mean and the variance of the stationary AoI. All the results are compared numerically, including the AoI performance of the non-relay state updating system. Numerical results show that adding the relay improves the system's timeliness dramatically, and \emph{no-buffer-and-preemption} setting in relay achieves both minimal average AoI and AoI's variance. Thus, for the system model discussed in this paper, to reduce the AoI at receiver there is no need to add the buffer in relay.
\end{abstract}



%
\IEEEpeerreviewmaketitle

\newtheorem{Definition}{Definition}
\newtheorem{Theorem}{Theorem}
\newtheorem{Lemma}{Lemma}
\newtheorem{Corollary}{Corollary}
\newtheorem{Proposition}{Proposition}
\newtheorem{Remark}{Remark}

\section{Introduction}

\IEEEPARstart{A}{part} from the traditional performance targets such as low propagation delay and large throughput, modern network designs put forward new requirements on information freshness and the timeliness of an information transmission system. In many remote monitoring system, the source has to transmit the messages or signals via a channel fast to keep the receiver's knowledge about certain physical process in source side fresh enough. These messages are time sensitive, whose freshness is of great importance in some applications. For example, the controller of a driverless car must obtain sufficiently new messages to keep the car running safely. The central node in a sensor network needs the fresh data from its neighbours to perform certain real time computations, such that an optimal resource scheduling policy is determined.

In general, this kind of information transmission system is abstracted into a status updating system, where the updates are generated in source and then sent to the receiver. The state at the receiver is refreshed after an updating packet is obtained successfully. To characterize the timeliness of the system or measure the freshness of an updating packets, a timeliness metric called age of information (AoI) was introduced in \cite{1}. Since then, lots of articles has been published focusing on different aspects of AoI, such as the AoI performance analysis and designing various optimal status updating system who has minimal average AoI. A detailed survey about AoI can be found in paper \cite{2}, in which the authors summarize the recent contributions of AoI in the aspect of theory and practice.

For simple queue models such as $M/M/1$, $M/D/1$, and $M/M/\infty$, the average AoI of state updating system were derived in papers \cite{3,4,5}. In particular, the Stochastic Hybrid System (SHS) analysis which is a systematic method to calculate the mean of AoI was introduced in paper \cite{5}. The method has also been used to determine the average AoI of some other status updating systems in \cite{6,7,8,9,10,11,12} recently. Packet managements problem was considered in \cite{13}, where the authors defined another metric called peak age of information (PAoI) and derived the average performance of both AoI and PAoI for three small-size updating systems. Assume that the packet has a fixed or random deadline, the long term average value of steady-state AoI was calculated in work \cite{14}. Apart from the analytical results of AoI described above, in practical applications the AoI has been considered widely in optimal system designs, either as the performance target or acting as one of constraints \cite{15,16,17,18,19,20,21,22}.

Observing that in the majority of systems considered before, an updating packet takes only one hop from source to the destination. Meanwhile, often there is only one path between the two nodes. However, both the packet multihop transmissions and the packet multipath transmissions are existent everywhere in practical communication networks, especially in wireless environment. The AoI of status updating system having two parallel servers was considered in \cite{23}, the authors computed the average AoI of this two-server system by sophisticate random events analysis. In paper \cite{24} the AoI in multihop networks was analyzed. Recently, the AoI of a three node status updating system was considered in \cite{25}, where a relay was added to cooperate the packet transmission from source to receiver. The authors showed that although two-hop transmission via the relay takes longer propagation delay, but the transmission success probability is raised, in particular when the direct line from source to destination is bad. Summarize both aspects, they proved that adding a relay can reduce the average AoI at the receiver, thus improving the system's timeliness.

In this paper, we consider a more general relay-assisted status updating system than that was discussed in \cite{25}, and derive more fundamental results about system's AoI. There, the situation that the relay has buffer was not mentioned and, for the non-relay cases, only the mean of steady-state AoI was calculated. We investigate both non-relay cases and the situation where the relay has a size 1 buffer. For each cae, the stationary distribution of AoI was derived explicitly, thus obtaining the complete description of the steady-state AoI. In \cite{25}, the authors derived their results using the most initial formula for the limiting time average AoI, which is used to deal with continuous time AoI. However, as the time is slotted, the AoI is discretized as well. Thus, the AoI considered in that paper, and in many other papers where the time was slotted is actually discrete. Here, we show that the analysis of AoI can be done by the idea and standard tools from queueing theory. The idea is invoking a multiple-dimensional discrete state vector which contains the AoI as a part, and constituting the multiple-dimensional AoI stochastic process. As long as the steady state of the constituted process is determined, i.e., all the stationary probabilities are solved, the stationary AoI distribution can be found by merging all the stationary probabilities with identical AoI state component, as it is in fact one of the marginal distributions of the larger-dimensional AoI process's steady-state distribution.

We point out that the idea that invoking a multiple-dimensional state vector to characterize the random transfers of AoI and analyzing the steady state of the constituted AoI process is also the idea of the SHS approach when the continuous AoI is discussed. Differently, finding the steady-state of a AoI stochastic hybrid system, one has to solve a system of differential equations which is in general very hard. While when the discrete AoI is analyzed, solving the system of stationary equations of constituted discrete AoI process is easier. Just recently, in paper \cite{26,27,28,29}, the authors proved that following above straightforward thinking, the stationary distribution of continuous AoI of status updating system with a few of simple queue models can also be determined, using AoI's Laplace Stieltjes Transform (LST) form. Finding the stationary distribution of AoI gets more attention in recent two or three years, both for continuous and discrete form of AoI \cite{30,31,32,33}.

The rest of the paper is organized as follows. We describe the model of a relay-assisted state updating system in Section II and introduce the system parameters. In Section III, assume that the relay has no buffer, we analyze the system's AoI of two cases where the packet retransmissions from relay can and cannot be preempted by later updates from the source node. The explicit expressions of steady-state AoI are derived by constituting a two-dimensional stochastic process. For the situation the relay is equipped with a size 1 buffer, the AoI of relay-assisted system is discussed in Section IV. There, a three-dimensional state vector is defined and the steady state of a three-dimensional AoI stochastic process is analyzed. The stationary AoI distribution is obtained as well. Provided the AoI distributions obtained before, for each case, we calculate the mean and the variance of AoI in Section V. Numerical results are given in Section VI, in which we compare the average AoI and the AoI's variance of all three cases. Besides, to demonstrate the improvement of adding the relay, we also calculate the AoI performance of an updating system without relay. The paper is summarized in Section VII, where some possible generalizations of relay-assisted system are discussed as well.

\section{System Model and Problem Formulation}
The system model is depicted in Figure \ref{fig1}. The source $s$ observes certain physical process and samples the state of the process at random times. The packet arrival is assumed to be a Bernoulli process, that is a packet arrives in each time slot independently with identical probability $p$. An updating packet is delivered to destination $d$, so that the knowledge of receiver about the process is refreshed. Since the existence of relay $R$, the packet can be transmitted either from the direct line or the two-hop link via $R$. Specifically, when the source $s$ transmits a packet, the same packet is also sent to $R$ at the same time.

\begin{figure}[!t]
\centering
\includegraphics[width=3.4in]{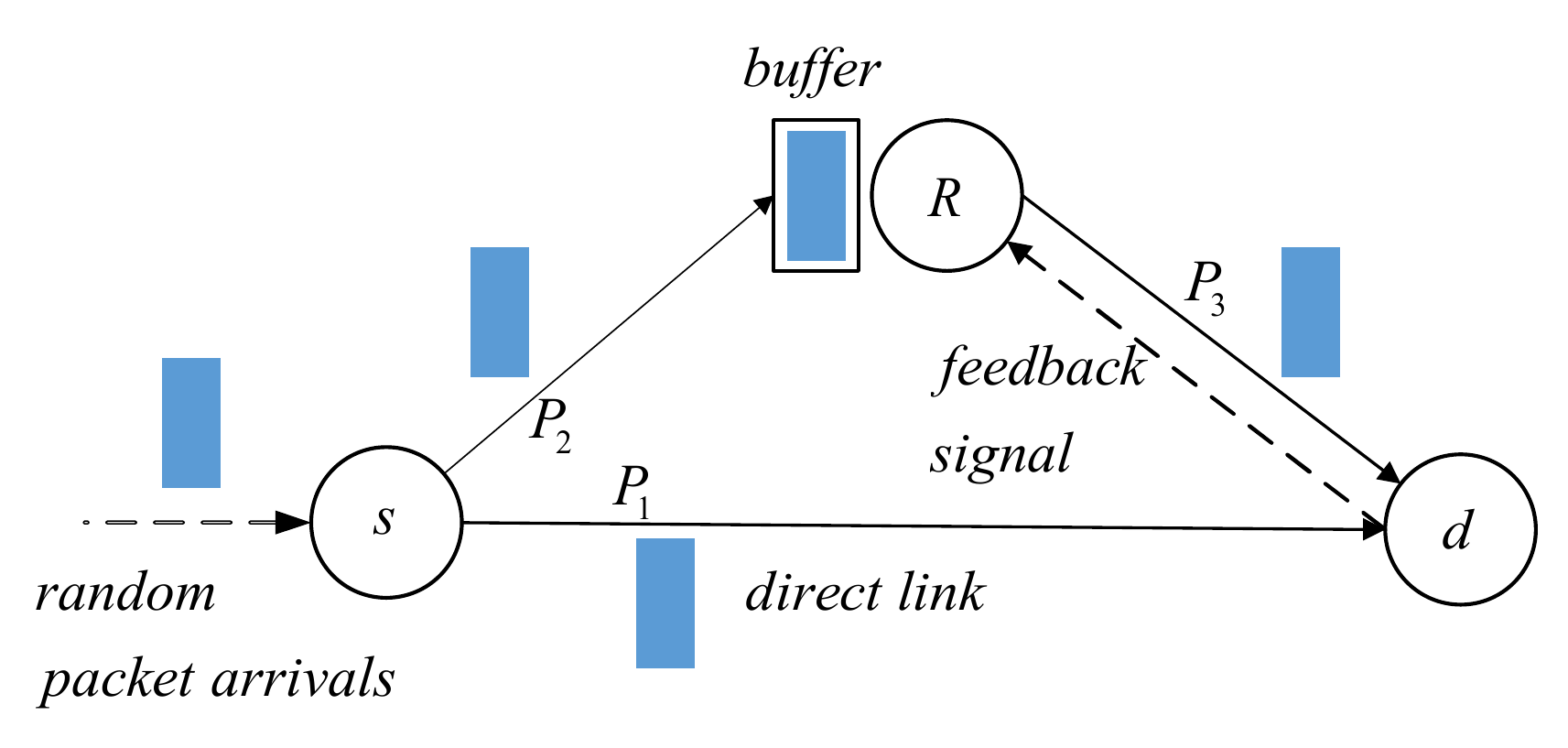}
\caption{The model of a status updating system with a relay.}
\label{fig1}
\end{figure}

Assume that $R$ gets a feedback signal when every time an updating packet is obtained at the receiver. In this case, the packet stored in $R$ is deleted to avoid redundant transmissions. In particular, if one packet is transmitted on direct link successfully, then all the packets in $R$ need to be deleted, including the one waiting in buffer when we consider the \emph{relay-has-buffer} case, since all of these packets are out of date. Otherwise, the relay will replace the source and retransmit the packet to $d$. Since $R$ is closer to $d$, it is more likely that the transmission is successful. Notice that the $s-R$ link is also unreliable, it is possible that both $s-R$ and $s-d$ transmissions fail at the same time. The source $s$ triggers the next transmission when a new packet arrives.
We are interested in the stationary AoI of the relay-assisted system for three cases, i.e., when the relay $R$ has no buffer, we discuss the AoI performance assuming that the packet retransmission can and cannot be preempted by fresher packets from source, and the steady-state AoI of the system when the relay is equipped with a size 1 buffer.

The three parameters $P_1$, $P_2$, and $P_3$ denote the transmission success probabilities of $s-d$, $s-R$, and $R-d$ links, respectively. Generally speaking, $P_1$ is less than $P_2$ and $P_3$. The propagation delays on all links are suppose to be one time slot. In this paper, we shall prove that for all of three cases the stationary distribution of AoI can be determined, such that the complete description of the steady-state AoI for the considered relay-assisted status updating system is obtained.

\section{Analysis of Age of Information: Relay has no buffer}
In this Section, assume that the relay has no buffer, we analyze the stationary AoI of the relay-assisted status updating system by constituting a two-dimensional discrete stochastic process. We determine the explicit expressions of the AoI distribution for both the cases where the packet in relay can and cannot be preempted by the newer ones from thee source.


Observing that if no packet arrives to the destination, the AoI increases by 1. The value of AoI is reduced to the age of the packet, which is actually equal to the system time of the packet until it is eventually received at receiver $d$. Define the two-dimensional discrete state vector $(n,m)$, where $n$ denotes the AoI at $d$, and $m$ represents the age of packet which is stored in relay $R$. When there is no packet in $R$, we let $m=0$.

In the following, for each case, we describe the random state transfers of each state vector, establish the stationary equations of the constituted AoI stochastic process and finally find its steady state, i.e., determining the stationary probabilities of every state vector $(n,m)$. Then, the stationary distribution of AoI is obtained by merging all the probabilities of those state vectors that have identical AoI-component $n$.

Notice that in the following paragraphs we also call the state vector as the age-state of the system, or simply the state.

\subsection{Analysis of AoI: packet preemption is allowable in relay}
Firstly, we analyze the AoI of the system assuming that the packet in $R$ is substituted when every time a fresher update is obtained from the source. The random transfers of age-state $(n,m)$ are described in the following.

First of all, assume that the age-state is $(n,m)$ where $n> m\geq1$. In this case, a packet of age at least 1 is contained in $R$. If no packet arrives to source at next time slot, then the state transfers are dependent on the packet transmission on $R-d$ link. With probability $P_3$, the packet delivery succeeds and the state vector changes to $(m+1,0)$. The second state-component is zero because the packet in $R$ is deleted when it was sent to $d$ successfully. Otherwise, both $n$ and $m$ will increase 1 at next time slot, i.e., the age-state turns to $(n+1,m+1)$. For the case a new packet is generated and sent out from source $s$, we have to analyze all the packet transmissions on links $s-R$, $R-d$, and direct link $s-d$.

Let $A$ be the random variable of packet arrivals, and we use $T_{s-d}$, $T_{s-R}$, $T_{R-d}$ to indicate whether the packet transmission is successful in three links, respectively. $A=0$ denotes no packet arrives and $T_{s-d}=0$ means that the packet transmission on $s-d$ link fails. Other variables are defined in similar manner. We list all the possible state transitions in Table \ref{table1}.

\begin{table}[!t]
\caption{Description of state vector transfers of AoI stochastic process $AoI_{NB}^P$}
\label{table1}
\centering
\begin{tabular}{c|c|c}
\Xhline{1pt}
Current state       &    $A=0,T_{R-d}$    &    $A=1,T_{s-d},T_{s-R},T_{R-d}$ \\
\Xhline{1pt}
\multirowcell{8}{$(n,m)$,\\$n>m\geq1$}    &  $(0,0):(n+1,m+1)$     &  $(1,0,0,0):(n+1,m+1)$\\
\cline{2-3} &  $(0,1):(m+1,0)$         &  $(1,0,0,1):(m+1,0)$\\
\cline{2-3} &                                    &$(1,0,1,0):(n+1,1)$\\
\cline{2-3} &                                    &$(1,0,1,1):(m+1,1)$\\
\cline{2-3} &                                    &$(1,1,0,0):(1,0)$\\
\cline{2-3} &                                    &$(1,1,0,1):(1,0)$\\
\cline{2-3} &                                    &$(1,1,1,0):(1,0)$\\
\cline{2-3} &                                    &$(1,1,1,1):(1,0)$\\
\Xhline{1pt}
Current state       &    $A=0$    &    $A=1,T_{s-d},T_{s-R}$ \\
\Xhline{1pt}
\multirowcell{4}{$(n,0)$,\\$n\geq1$}    &  $(n+1,0)$     &  $(1,0,0):(n+1,0)$\\
\cline{2-3} &                                    &  $(1,0,1):(n+1,1)$\\
\cline{2-3} &                                    &$(1,1,0):(1,0)$\\
\cline{2-3} &                                    &$(1,1,1):(1,0)$\\
\Xhline{1pt}
\end{tabular}
\end{table}

\begin{table}[!t]
\caption{All the state transfers and corresponding transition probabilities of stochastic process $AoI_{NB}^P$}
\label{table3}
\centering
\begin{tabular}{c|c}
\Xhline{1pt}
From       & State transfers to at next slot/Transition probabilities \\
\Xhline{1pt}
\multirowcell{5}{$(n,m)$,\\$m\geq1$} & $(n+1,m+1):(1-p)(1-P_3)+ p(1-P_1)(1-P_2)(1-P_3)$\\
\cline{2-2} &  $(m+1,0):(1-p)P_3+p(1-P_1)(1-P_2)P_3$\\
\cline{2-2} &  $(n+1,1):p(1-P_1)P_2(1-P_3)$\\
\cline{2-2} &  $(m+1,1):p(1-P_1)P_2P_3$\\
\cline{2-2} &  $(1,0):pP_1$\\
\hline
\multirowcell{3}{$(n,0)$,\\$n\geq1$} & $(n+1,0):(1-p)+p(1-P_1)(1-P_2)$\\
\cline{2-2} & $(n+1,1):p(1-P_1)P_2$\\
\cline{2-2} & $(1,0):pP_1$\\
\Xhline{1pt}
\end{tabular}
\end{table}

It is necessary to explain all the cases including $A=1$ and $T_{s-d}=1$ in Table \ref{table1}. When a new packet is delivered successfully over direct link $s-d$, the value of AoI must reduce to 1 since a packet of age 1 is obtained in $d$, although it is possible that an another packet is received on $R-d$ link. At the same time, according to our assumption, a feedback signal is sent to $R$ and the relay deletes the packet received at current time slot or before. Thus, $m$ is reduced to zero and the state vector jumps to $(1,0)$.

Still we have to consider the state transfers when the initial age-state is $(n,0)$, which means the relay $R$ is currently empty. If no packet arrives, it is easy to see that the age-state transfers to $(n+1,0)$ at next time slot. On the contrary, for the case $A=1$, we need to consider the packet transmissions on $s-d$ and $s-R$ links. All the possible state transitions from age-state $(n,0)$ are given in the last four lines in Table \ref{table1}.

Constituting the two-dimensional AoI stochastic process
$$
AoI_{NB}^P=\{(n_k,m_k):n_k> m_k\geq0,k\in\mathbb{N}\},
$$
where the subscript $NB$ denotes that at relay $R$ there is no buffer, and the superscript $P$ indicates that the packet retransmission in $R$ can be preempted. 
Define $P_{(n,m),(i,j)}$ be the single step transition probability from $(n,m)$ to age-state $(i,j)$. Summarize above discussions, we list all the state transfers and corresponding transition probabilities of the AoI process $AoI_{NB}^P$ in Table \ref{table3}.
Denote that $\pi_{(n,m)}$ is the probability of age-state $(n,m)$, $n> m\geq0$ when the AoI process runs in steady state. In following Theorem 1, we establish the stationary equations of the process $AoI_{NB}^P$.

\begin{Theorem}
Assume that the stochastic process $AoI_{NB}^P$ reaches the steady state, then the stationary probabilities $\pi_{(n,m)}$ of age-state $(n,m)$, $n>m\geq0$ satisfy that
\begin{equation}
\begin{cases}
\pi_{(n,m)}=\pi_{(n-1,m-1)}\big[(1-p)(1-P_3) + p(1-P_1)(1-P_2)(1-P_3)\big]            & \quad   (n>m\geq2)  \\
\pi_{(n,1)}=\pi_{(n-1,0)}p(1-P_1)P_2 + \sum\nolimits_{j=1}^{n-2}\pi_{(n-1,j)}p(1-P_1)P_2(1-P_3)  \\
\qquad \qquad \qquad  + \sum\nolimits_{k=n}^{\infty}\pi_{(k,n-1)}p(1-P_1)P_2P_3                     &  \quad   (n\geq3)  \\
\pi_{(2,1)}=\pi_{(1,0)}p(1-P_1)P_2 + \sum\nolimits_{k=2}^{\infty}\pi_{(k,1)}p(1-P_1)P_2P_3   \\
\pi_{(n,0)}=\pi_{(n-1,0)}\left[(1-p)+p(1-P_1)(1-P_2)\right]    \\
\qquad \qquad \qquad  + \sum\nolimits_{k=n}^{\infty}\pi_{(k,n-1)}\left[(1-p)P_3+p(1-P_1)(1-P_2)P_3\right]   & \quad   (n\geq2)  \\
\pi_{(1,0)}=\left(\sum\nolimits_{n=1}^{\infty}\pi_{(n,0)}+ \sum\nolimits_{m=1}^{\infty}\sum\nolimits_{n=m+1}^{\infty}\pi_{(n,m)}   \right) pP_1
\end{cases}
\end{equation}
\end{Theorem}

\begin{proof}
Equations in (1) are explained as follows. First of all, begin with $(n-1,m-1)$, as long as the packet transmission on $R-d$ link fails which occurs with probability $1-P_3$, along with the condition that either no packet is delivered on $s-R$ and $s-d$ links or both transmissions are unsuccessful, then the age-state transfers to $(n,m)$ at next time slot. This gives the first line of equations (1).

To obtain $(n,1)$, more cases need to be considered. From age-state $(n-1,0)$ which means $R$ is initially empty, at source $s$ a packet arrives and is transmitted to $R$ successfully makes the second parameter jump to 1. On the other hand, let the packet delivery on $s-d$ link fail, such that the AoI in $d$ increases to $n$ at next time slot. This state transitions occurs with probability $p(1-P_1)P_2$. Next, notice that the packet in $R$ is allowed to be refreshed, the second state-component $1$ can also be created by packet replacement. Assume that the initial state vector is $(n-1,j)$, $1\leq j\leq n-2$. When the new packet is transmitted on $s-R$ link successfully but the packet transmissions on both $s-d$ and $R-d$ links fail, it was observed that the state also transfers to $(n,1)$. Thus, we obtain the second term of the RHS of second equation. Now, suppose that a packet of age $n-1$ goes through $R-d$ link to the receiver. As a result, the AoI in $d$ is reduced to $n$ at next time slot. At the same time, the relay $R$ obtains a new packet from $s$ via $s-R$ link, which changes the second parameter to 1. In this case, we can obtain the age-state $(n,1)$ again. The initial state can be $(k,n-1)$ where $k$ is an arbitrary number greater than $n-1$. For the case $n=2$, packet substitution in $R$ is nonexistent since $R$ is empty before one time slot. With all above discussions, we explain how the state vector $(n,1)$, $n\geq2$ is obtained and determine its stationary equation.

Then, the random transitions to $(n,0)$ is analyzed. The relay $R$ is empty either it is empty initially, or is emptied as the contained packet is sent to $d$ successfully. Therefore, two cases have to be analyzed. Firstly, from age-state $(n-1,0)$, as long as no packet arrives to $R$ and $d$, obviously the state jumps to $(n,0)$ after one time slot. On the other hand, let the initial state is $(k,n-1)$, $k\geq n$, a packet of age $n-1$ can be transmitted to $d$ via $R-d$ link and assume that no other packet is obtained from $s$ through the direct link $s-d$, which ensures that the AoI is set to $n$ at next time slot. The second parameter maintains 0 if no packet comes to $s$, or the arriving packet is not transmitted to $R$ successfully on $s-R$ link. Combined with the transition probabilities in Table \ref{table3}, we give the stationary equation for the state $(n,0)$ in the end.

The only way to get age-state $(1,0)$ is transmitting a packet via the direct link $s-d$, no matter what the initial state is. Actually, this last equation determines $\pi_{(1,0)}=pP_1$ directly, since the probabilities in bracket add up to 1.

So far, all the state transfers are considered, and the proof of Theorem 1 completes.
\end{proof}

We solve the system of equations (1) and determine all the stationary probabilities $\pi_{(n,m)}$ in Theorem 2, whose proof is postponed to Appendix A.
\begin{Theorem}
The probabilities $\pi_{(n,m)}$ of AoI process $AoI_{NB}^P$ are determined as follows. Firstly,
\begin{equation}
\pi_{(n,0)}=(1-\delta)\delta^{n-1} - p(1-P_1)P_2[(1-P_3)\delta]^{n-1}  \qquad (n\geq1)
\end{equation}
and for $n>m\geq1$
\begin{align}
\pi_{(n,m)}&=\pi_{(2,1)}[(1-pP_1)(1-P_3)]^{n-2}\left(\frac{\delta}{1-pP_1}\right)^{m-1} + \frac{p(1-P_1)P_2P_3(1-\delta)\delta}{\delta-(1-pP_1)(1-P_3)} \notag \\
& \qquad \quad \times \left[ \delta^{n-2}(1-P_3)^{m-1} - [(1-pP_1)(1-P_3)]^{n-2}\left(\frac{\delta}{1-pP_1} \right)^{m-1} \right]
\end{align}
where
\begin{align}
\delta=1-p(P_1+P_2-P_1P_2),  \qquad  \pi_{(2,1)}=p^2(1-P_1)P_2[P_1+(1-P_1)P_2P_3]  \notag
\end{align}
\end{Theorem}

Provided all the stationary probabilities, the probability that the steady-state AoI equals $n$ can be obtained by
\begin{equation}
\Pr\{\Delta=n\}=\pi_{(n,0)} + \sum\nolimits_{m=1}^{n-1}\pi_{(n,m)} \qquad (n\geq2)
\end{equation}

The probability that AoI takes value 1 is determined as $pP_1$, which is known directly from the last line of equations (1). We determine the explicit expression of AoI distribution in following Theorem 3 and give the calculations in Appendix B.

\begin{Theorem}
For the considered relay-assisted status updating system, the stationary distribution of AoI $\Delta_{NB}^P$ is determined as
\begin{align}
\Pr\{\Delta_{NB}^P=n\}=\frac{(1-pP_1)P_3(1-\delta)}{(1-pP_1)P_3-p(1-P_1)P_2}\delta^{n-1} +\xi[(1-pP_1)(1-P_3)]^{n-1}  \quad   (n\geq1)
\end{align}
in which
\begin{align}
\xi=\frac{p[P_1+(1-P_1)P_2P_3]}{1-P_3} - \frac{P_3(1-\delta)\delta}{[(1-pP_1)P_3-p(1-P_1)P_2](1-P_3)}  \notag
\end{align}
\end{Theorem}

\subsection{The stationary AoI: packet retransmission is not preempted in relay}

\begin{table}[!t]
\caption{Possible age-state transitions when there is no packet preemption in $R$}
\label{table4}
\centering
\begin{tabular}{c|c|c}
\Xhline{1pt}
Current state       &    $A=0,T_{R-d}$    &    $A=1,T_{s-d},T_{R-d}$ \\
\Xhline{1pt}
\multirowcell{4}{$(n,m)$,\\$n>m\geq1$}    &  $(0,0):(n+1,m+1)$     &  $(1,0,0):(n+1,m+1)$\\
\cline{2-3} &  $(0,1):(m+1,0)$         &  $(1,0,1):(m+1,0)$\\
\cline{2-3} &                                    &$(1,1,0):(1,0)$\\
\cline{2-3} &                                    &$(1,1,1):(1,0)$\\
\Xhline{1pt}
Current state       &    $A=0$    &    $A=1,T_{s-d},T_{s-R}$ \\
\Xhline{1pt}
\multirowcell{4}{$(n,0)$,\\$n\geq1$}    &  $(n+1,0)$     &  $(1,0,0):(n+1,0)$\\
\cline{2-3} &                                    &  $(1,0,1):(n+1,1)$\\
\cline{2-3} &                                    &$(1,1,0):(1,0)$\\
\cline{2-3} &                                    &$(1,1,1):(1,0)$\\
\Xhline{1pt}
\end{tabular}
\end{table}

In previous part of this Section, we have determined the stationary AoI distribution for the case where the packet in relay can be replaced. Intuitively, substituting the large-age packet with fresher ones will result to lower AoI and enhance the timeliness of the system. We will prove this conclusion by finding the steady-state AoI distribution for the $non-preemption$ case explicitly and do the comparisons. The AoI stochastic process constituted for this case is denoted as $AoI_{NB}^{NP}$, where the superscript $NP$ indicates that no packet preemption is allowed in relay $R$.

When the packet in relay cannot be preempted before the receiver gets an updating packet eventually, which may be obtained from $R-d$ link or the $s-d$ link, the discussions of random state transitions are easier, since the relay will always reject the packet from the source as long as there has been one packet within it.

We first describe the possible state transitions when the initial age-state is $(n,m)$, where $n>m\geq1$. If $A=1$, that is a new packet is sent out from the source in this time slot. Now, only the packet transmissions on $s-d$ and $R-d$ links need to be discussed, since the packet over $s-R$ link is always rejected. For the case $A=0$, we have to consider just the packet delivery over $R-d$ link. All the possible state transfers are listed in Table \ref{table4}, in which the state transfers from initial age-state $(n,0)$ are also given.

For each state transfer, we determine the corresponding transition probability in Table \ref{table5}, from which the stationary equations of the AoI process $AoI_{NB}^{NP}$ can be established, when the process reaches the steady state.

\begin{Theorem}
Assume that the AoI stochastic process $AoI_{NB}^{NP}$ reaches the steady-state, the stationary probabilities $\pi_{(n,m)}$ satisfy following equations
\begin{equation}
\begin{cases}
\pi_{(n,m)}=\pi_{(n-1,m-1)}[(1-p)(1-P_3)+p(1-P_1)(1-P_3)]   &    (n>m\geq2)    \\
\pi_{(n,1)}=\pi_{(n-1,0)}p(1-P_1)P_2                                        &   (n\geq2)   \\
\pi_{(n,0)}=\pi_{(n-1,0)}[(1-p)+p(1-P_1)(1-P_2)]    \\
\qquad \quad \quad      + \left(\sum\nolimits_{k=n}^{\infty}\pi_{(k,n-1)}\right)[(1-p)P_3+p(1-P_1)P_3]    &   (n\geq2)    \\
\pi_{(1,0)}= \left(\sum\nolimits_{n=1}^{\infty}\pi_{(n,0)} + \sum\nolimits_{n=2}^{\infty}\sum\nolimits_{m=1}^{n-1}\pi_{(n,m)}\right)pP_1
\end{cases}
\end{equation}
\end{Theorem}

\begin{proof}
Assume that the initial state is $(n-1,m-1)$, to get age-state $(n,m)$, it needs that no packet is obtained in both relay and destination, which is satisfied either no packet arrives at source and retransmission from relay fails, or a new packet comes to source but the packet transmissions on $s-d$ and $R-d$ links are unsuccessful. Since the packet in relay cannot be substituted, the state $(n,1)$ can only be obtained from $(n-1,0)$ by letting a new packet arrive to relay $R$ but fail to transmit to receiver $d$. This gives the first two lines of equations (6).

Similarly to the explanations for the second to last equation in (1), both the state $(n-1,0)$ and $(k,n-1)$ can transfer to $(n,0)$. Beginning with $(n-1,0)$, if no packet comes to source or the arriving packet is not sent to $R$ and $d$ successfully, no packet is obtained at $R$ or $d$. As a result, at the next time slot the AoI increases 1 while the second parameter remains 0. On the other hand, a packet of age $n-1$ can be sent to $d$ from $R$ which reduces the AoI to $n$. Meanwhile, still we have to ensure that the receiver gets no packet from the direct link $s-d$. Combining these two cases, we determine the stationary equation of the state vector $(n,0)$ where $n\geq2$.

Finally, it sees that still only one way can decrease the AoI to 1, i.e., the arriving packet is transmitted to $d$ through the direct link $s-d$, no matter what the original age-state is. As before, this equation yields that the stationary probability $\pi_{(1,0)}$ is $pP_1$ directly.

So far, we have explained all the equations in (6). This completes the proof of Theorem 4.
\end{proof}

\begin{table}[!t]
\caption{Random state transfers and transtion probabilities: non-preemption case}
\label{table5}
\centering
\begin{tabular}{c|c}
\Xhline{1pt}
From       & State transfers to at next slot/Transition probabilities \\
\Xhline{1pt}
\multirowcell{3}{$(n,m)$,\\$n>m\geq1$} & $(n+1,m+1):(1-p)(1-P_3) + p(1-P_1)(1-P_3)$\\
\cline{2-2} &  $(m+1,0):(1-p)P_3+p(1-P_1)P_3$\\
\cline{2-2} &  $(1,0):pP_1$\\
\hline
\multirowcell{3}{$(n,0)$,\\$n\geq1$} & $(n+1,0):(1-p)+p(1-P_1)(1-P_2)$\\
\cline{2-2} & $(n+1,1):p(1-P_1)P_2$\\
\cline{2-2} & $(1,0):pP_1$\\
\Xhline{1pt}
\end{tabular}
\end{table}

Solving the system of equations (6) is easier than (1), since \emph{no-packet-preemption} in relay $R$ dramatically simplifies how the age-state $(n,1)$ is obtained. We give the solutions of all the stationary probabilities $\pi_{(n,m)}$ in Theorem 5 and put the calculation details in Appendix C.

\begin{Theorem}
When the AoI process $AoI_{NB}^{NP}$ reaches the steady-state, the stationary probabilities $\pi_{(n,m)}$ are solved as
\begin{equation}
\pi_{(n,0)}=\beta_1\delta^{n-1} + \beta_2[(1-pP_1)(1-P_3)]^{n-1} \quad (n\geq1)
\end{equation}

Other probabilities $\pi_{(n,m)}$, $n>m\geq1$ are determined by
\begin{align}
\pi_{(n,m)}&=p(1-P_1)P_2\bigg\{\beta_1\delta^{n-2}\left(\frac{(1-pP_1)(1-P_3)^{m-1}}{\delta}\right) + \beta_2\left[ (1-pP_1)(1-P_3) \right]^{n-2}\bigg\},
\end{align}
in which the coefficients $\delta$, $\beta_1$ are $\beta_2$ are given by
\begin{align}
\delta&=1-p(P_1 + P_2 - P_1P_2),  \notag  \\  
\beta_1&= pP_1-\beta_2,   \\
\beta_2&= \frac{[1-(1-pP_1)(1-P_3)]p(1-P_1)P_2P_3}{[1-(1-pP_1)(1-P_3)+p(1-P_1)P_2]} \frac{1-pP_1}{(1-pP_1)(1-P_3)-\delta}.
\end{align}
\end{Theorem}

Now, the stationary distribution of AoI, which we denote as $\Delta_{NB}^{NP}$ can be determined.
\begin{Theorem}
Assume that the packet in relay cannot be substituted by later packets from source, then for $n\geq1$, the stationary AoI distribution of relay-assisted status updating system is given as
\begin{align}
\Pr\{\Delta_{NB}^{NP}=n\}&=\left(\beta_1 + \frac{p(1-P_1)P_2}{\delta-(1-pP_1)(1-P_3)}\beta_1\right) \delta^{n-1}  \notag  \\
& \quad + \left(\beta_2 - \frac{p(1-P_1)P_2}{\delta-(1-pP_1)(1-P_3)}\beta_1\right) [(1-pP_1)(1-P_3)]^{n-1}  \notag \\
& \quad + p(1-P_1)P_2\beta_2(n-1)[(1-pP_1)(1-P_3)]^{n-2}
\end{align}
\end{Theorem}

\begin{proof}
As before, for $n\geq2$, we first calculate the sum
\begin{align}
&\sum\nolimits_{m=1}^{n-1}\pi_{(n,m)}  \notag \\
={}& \sum\nolimits_{m=1}^{n-1} p(1-P_1)P_2\bigg\{\beta_1\delta^{n-2}\left(\frac{(1-pP_1)(1-P_3)^{m-1}}{\delta}\right) + \beta_2\left[ (1-pP_1)(1-P_3) \right]^{n-2}\bigg\}   \notag \\
={}& \frac{p(1-P_1)P_2\beta_1}{\delta-(1-pP_1)(1-P_3)}\left\{\delta^{n-1}-[(1-pP_1)(1-P_3)]^{n-1}\right\}  \notag \\
  {}& \qquad + p(1-P_1)P_2\beta_2(n-1)[(1-pP_1)(1-P_3)]^{n-2}
\end{align}

Thus, we have that
\begin{align}
&\Pr\{\Delta_{NB}^{NP}=n\}  \notag \\
={}& \pi_{(n,0)} + \sum\nolimits_{m=1}^{n-1}\pi_{(n,m)}  \notag \\
={}& \left(\beta_1 + \frac{p(1-P_1)P_2}{\delta-(1-pP_1)(1-P_3)}\beta_1\right) \delta^{n-1} + \left(\beta_2 - \frac{p(1-P_1)P_2}{\delta-(1-pP_1)(1-P_3)}\beta_1\right)  [(1-pP_1)(1-P_3)]^{n-1}  \notag \\
  {}& \qquad + p(1-P_1)P_2\beta_2(n-1)[(1-pP_1)(1-P_3)]^{n-2}
\end{align}
where the probability expression (7) is substituted.

Let $n=1$, equation (13) reduces to $\beta_1 + \beta_2=pP_1$ from (9), which equals $\pi_{(1,0)}$ and is also the probability that AoI takes value 1. Therefore, (11) is valid for all $n\geq1$. This completes the proof of Theorem 6.
\end{proof}

\section{Characterizing the Age of Information When relay has size 1 buffer}
In the previous Section, we have discussed the relay-assisted system where the relay has no buffer. The explicit expressions of stationary AoI distribution was derived for two cases in which we assume that the new packet can and cannot preempt the packet retransmissions from relay. In this Section, we will consider the steady-state AoI when the relay is equipped with a size 1 buffer, which is used to store the arriving packet if there is already one packet in relay. In addition, the packet in buffer is allowed to be updated by later ones from source $s$.

The purpose of this Section is two-fold. On the one hand, we will prove that the AoI distribution can still be determined by constituting a multiple-dimensional stochastic process and solving its steady state. Hence, not only the mean of AoI, but all the AoI performances can be computed. On the other hand, finding out the influences of \emph{storing-another-packet} on the system's AoI is practically meaningful, which directly guides the design of a superior statue updating system.

\begin{table}[!t]
\caption{The age-state transfers of relay-assisted status updating system: relay has size 1 buffer}
\label{table6}
\centering
\begin{tabular}{c|c|c}
\Xhline{1pt}
Current state vector       &     $A=0,T_{R-d}$    &    $A=1,T_{s-d},T_{s-R},T_{R-d}$ \\
\Xhline{1pt}
\multirowcell{8}{$(n,m,l)$,\\$n>m>l\geq1$}    &  $(0,0):(n+1,m+1,l+1)$     &  $(1,0,0,0):(n+1,m+1,l+1)$\\
\cline{2-3} &  $(0,1):(m+1,l+1,0)$         &  $(1,0,0,1):(m+1,l+1,0)$\\
\cline{2-3} &                                           & $(1,0,1,0):(n+1,m+1,1)$\\
\cline{2-3} &                                           & $(1,0,1,1):(m+1,1,0)$\\
\cline{2-3} &                                           & $(1,1,0,0):(1,0,0)$\\
\cline{2-3} &                                           & $(1,1,0,1):(1,0,0)$\\
\cline{2-3} &                                           & $(1,1,1,0):(1,0,0)$\\
\cline{2-3} &                                           & $(1,1,1,1):(1,0,0)$\\
\Xhline{1pt}
Current state vector        &    $A=0,T_{R-d}$     &    $A=1,T_{s-d},T_{s-R},T_{R-d}$ \\
\Xhline{1pt}
\multirowcell{8}{$(n,m,0)$,\\$n>m\geq1$}    &  $(0,0):(n+1,m+1,0)$     &  $(1,0,0,0):(n+1,m+1,0)$\\
\cline{2-3} &   $(0,1):(m+1,0,0)$           &  $(1,0,0,1):(m+1,0,0)$\\
\cline{2-3} &                                           &$(1,0,1,0):(n+1,m+1,1)$\\
\cline{2-3} &                                           &$(1,0,1,1):(m+1,1,0)$\\
\cline{2-3} &                                           &  $(1,1,0,0):(1,0,0)$\\
\cline{2-3} &                                           &$(1,1,0,1):(1,0,0)$\\
\cline{2-3} &                                           &$(1,1,1,0):(1,0,0)$\\
\cline{2-3} &                                           &  $(1,1,1,1):(1,0,0)$\\
\Xhline{1pt}
Current state vector        &    $A=0$    &    $A=1,T_{s-d},T_{s-R}$ \\
\Xhline{1pt}
\multirowcell{4}{$(n,0,0)$,\\$n\geq1$}    &  $(n+1,0,0)$     &  $(1,0,0):(n+1,0,0)$\\
\cline{2-3} &                                    &  $(1,0,1):(n+1,1,0)$\\
\cline{2-3} &                                    &$(1,1,0):(1,0,0)$\\
\cline{2-3} &                                    &$(1,1,1):(1,0,0)$\\
\Xhline{1pt}
\end{tabular}
\end{table}

Here, a three-dimensional state vector $(n,m,l)$ is defined, and the three-dimensional stochastic process $AoI_B^{P}$ is constituted accordingly. The added parameter $l$ denotes the age of packet in relay's buffer. To describe the age-state transfers, again we list all the cases using a Table \ref{table6}.

\begin{table}[!t]
\caption{Random state transfers and the corresponding transtion probabilities: relay has size 1 buffer}
\label{table7}
\centering
\begin{tabular}{c|c}
\Xhline{1pt}
Initial age-state      &    State vector at next time slot/Transition probabilities \\
\Xhline{1pt}
\multirowcell{5}{$(n,m,l)$,\\$n>m>l\geq1$}      &      $(n+1,m+1,l+1):(1-p)(1-P_3) + p(1-P_1)(1-P_2)(1-P_3)$\\
\cline{2-2}    &    $(m+1,l+1,0):(1-p)P_3+p(1-P_1)(1-P_2)P_3$\\
\cline{2-2}    &    $(n+1,m+1,1):p(1-P_1)P_2(1-P_3)$\\
\cline{2-2}    &    $(m+1,1,0):p(1-P_1)P_2P_3$\\
\cline{2-2}    &    $(1,0,0):pP_1$\\
\hline
\multirowcell{5}{$(n,m,0)$,\\$n>m\geq1$}        &      $(n+1,m+1,0):(1-p)(1-P_3)+p(1-P_1)(1-P_2)(1-P_3)$\\
\cline{2-2}    &   $(m+1,0,0):(1-p)P_3+p(1-P_1)(1-P_2)P_3$\\
\cline{2-2}    &   $(n+1,m+1,1):p(1-P_1)P_2(1-P_3)$\\
\cline{2-2}    &   $(m+1,1,0):p(1-P_1)P_2P_3$\\
\cline{2-2}    &   $(1,0,0):pP_1$\\
\hline
\multirowcell{3}{$(n,0,0)$,\\$n\geq1$}              &       $(n+1,0,0):(1-p)+p(1-P_1)(1-P_2)$\\
\cline{2-2}    &   $(n+1,1,0):p(1-P_1)P_2$\\
\cline{2-2}    &   $(1,0,0):pP_1$\\
\Xhline{1pt}
\end{tabular}
\end{table}

The state transfers are described assuming that the process $AoI_B^{P}$ has different initial state vector, i.e., $(n,m,l)$, $(n,m,0)$ and $(n,0,0)$. Actually, we list all the state transitions in the order of number of packets in system. For instance, age-state $(n,m,l)$ means that the current value of AoI is $n$ and there are two packets of age $m$ and $l$ in relay, one is transmitted on $R-d$ link and the other is waiting in buffer. At the next time slot, if $A=0$, that is no packet is obtained at source, only the packet transmission over $R-d$ link needs to be considered, since no packet is delivered on both $s-d$ and $s-R$ links. When $A=1$, we have to describe whether the new packet is successfully transmitted via $s-d$ and $s-R$ links and in relay $R$, if the packet having age $m$ is obtained at destination $d$. This gives first eight rows of Table \ref{table6}. All the state tranfers from other initial age-states $(n,m,0)$ and $(n,0,0)$ are discussed similarly.

Notice that still only one way can reduce the AoI to 1, that is a new updating packet is sent to $d$ via the direct link $s-d$. Since we have assumed that when everytime a packet is obtained from $s-d$ link, all the packets in system are deleted, which makes the other two state components $m$ and $l$ jump to 0. Therefore, we can determine that the first stationary probability $\pi_{(1,0,0)}$ is equal to $pP_1$ immediately.

According to Table \ref{table6}, in Table \ref{table7} we summarize every state transfer and determine the transition probability, from which the stationary equations of AoI stochastic process $AoI_B^P$ can be established when the process reaches the steady state.

For the sake of simplicity, in this Section we denote that
\begin{align}
\delta=(1-p)+p(1-P_1)(1-P_2) =1-p(P_1+P_2-P_1P_2)  \notag
\end{align}
and
\begin{equation}
\eta=p(1-P_1)P_2   \notag
\end{equation}

\begin{Theorem}
Assume that the three-dimensional stochastic process $AoI_B^P$ runs in steady state, the stationary probabilities $\pi_{(n,m,l)}$ of all the age-states satisfy following equations
\begin{equation}
\begin{cases}
\pi_{(n,m,l)}=\pi_{(n-1,m-1,l-1)}\delta(1-P_3)     &  \qquad   (n>m>l\geq2)    \\
\pi_{(n,m,1)}=\left(\sum\nolimits_{j=0}^{m-2}\pi_{(n-1,m-1,j)}\right)\eta(1-P_3)  & \qquad   (n>m\geq2)    \\
\pi_{(n,m,0)}=\pi_{(n-1,m-1,0)}\delta(1-P_3) +\left(\sum\nolimits_{k=n}^{\infty}\pi_{(k,n-1,m-1)}\right)\delta P_3  & \qquad  (n>m\geq2)  \\
\pi_{(n,1,0)}=\pi_{(n-1,0,0)}\eta+\left(\sum\nolimits_{k=n}^{\infty}\sum\nolimits_{j=0}^{n-2}\pi_{(k,n-1,j)}\right)\eta P_3 & \qquad (n\geq2)    \\
\pi_{(n,0,0)}=\pi_{(n-1,0,0)}\delta+\left(\sum\nolimits_{k=n}^{\infty}\pi_{(k,n-1,0)}\right)\delta P_3  &   \qquad (n\geq2)  \\
\pi_{(1,0,0)}=\Big(\sum\nolimits_{n=1}^{\infty}\pi_{(n,0,0)}+ \sum\nolimits_{n=2}^{\infty}\sum\nolimits_{m=1}^{n-1}\pi_{(n,m,0)}\\
\qquad \qquad \qquad \qquad \qquad + \sum\nolimits_{n=3}^{\infty}\sum\nolimits_{m=2}^{n-1}\sum\nolimits_{l=1}^{m-1}\pi_{(n,m,l)}\Big)pP_1
\end{cases}
\end{equation}
\end{Theorem}

\begin{proof}
We explain each line of equations (14) briefly. For the cases $n>m>l\geq2$, the state $(n,m,l)$ can only be obtained from $(n-1,m-1,l-1)$ when both $d$ and $R$ obtains no packet in the current time slot. No packet comes to $d$ ensures that the AoI does not drop, while the age of packet stored in relay's buffer also increases 1 as long as no packet is sent to $R$ from $s$ successfully.

Since the packet contained in relay's buffer can be replaced by later updating packets, apart from $(n-1,m-1,0)$, state vectors of form $(n-1,m-1,j)$ where $1\leq j\leq m-2$ can also transfer to $(n,m,1)$. These state transfers occur if the AoI at $d$ does not decrease, and the packet in relay's buffer is updated by a new packet at the same time. This gives the second line of (14).

Age-state $(n,m,0)$ means that the buffer in $R$ is empty, which can be transferred to in two ways. At first, if no packet arrives to $R$ and $d$, then state vector $(n-1,m-1,0)$ transfers to $(n,m,0)$ after one time slot naturally. Apart from this, notice that when the packet retransmission is successful via $R-d$ link, the AoI is reduced and this packet is then deleted from $R$. Next time the packet in buffer is transmitted. As a result, any state having form $(k,n-1,m-1)$, $k\geq n$ will change to $(n,m,0)$ as well. Combining both cases, we obtain the third equation in (14). When considering the state $(n,1,0)$, remember that the packet in buffer can be replaced, making the last parameter $l$ reduce to 1. Similar to above discussions, we see that all the age-states of form $(k,n-1,j)$ where $k\geq n$ and $1\leq j\leq n-2$ can transfer to $(n,1,0)$, so that in this case we have a double sum in the stationary equation.

The empty state $(n,0,0)$ can be obtained from either $(n-1,0,0)$ when no packet arrives to $d$ and $R$, or age-state $(k,n-1,0)$, $k\geq n$ assuming that the packet of age $n-1$ is sent to $d$ through $R-d$ link and no other packet comes to relay $R$. Substituting the transition probabilities in Table \ref{table7}, the stationary equation corresponding to $(n,0,0)$ can be determined. At last, the AoI is reduced to 1 from any initial age-state when a new packet is delivered to $d$ through the direct link $s-d$, which yields the last line of equations (14) and in fact determines that $\pi_{(1,0,0)}=pP_1$.

This completes the proof of Theorem 7.
\end{proof}

Compared with the \emph{no-buffer} cases discussed before, here we have to find all the stationary probabilities $\pi_{(n,m,l)}$, $n>m>l\geq0$. Then, the probability that stationary AoI, $\Delta_B^P$, takes value $n$ is determined by the formula
\begin{align}
\Pr\{\Delta_B^P=n\}=\pi_{(n,0,0)} &+\sum\nolimits_{m=1}^{n-1}\pi_{(n,m,0)} +\sum\nolimits_{m=2}^{n-1}\sum\nolimits_{l=1}^{m-1}\pi_{(n,m,l)}
\end{align}

\begin{Theorem}
The solutions of equations (14), i.e., all the stationary probabilities $\pi_{(n,m,l)}$ of three-dimensional stochastic process $AoI_B^P$ are given as follows. First of all,
\begin{align}
\pi_{(n,0,0)}&=\left(pP_1+\eta\right)\delta^{n-1} -\left\{\eta-(1-S)\eta P_3\right\}[\delta(1-P_3)]^{n-1}  \notag \\
& \qquad \qquad \qquad \qquad   -(1-S)\eta P_3 n[\delta(1-P_3)]^{n-1}  \qquad  (n\geq1)
\end{align}

For $n>m\geq0$,
\begin{align}
&\pi_{(n,m,0)}  \notag \\
={}&\eta\left(pP_1+\eta\right)\delta^{n-2}\left(1-P_3\right)^{m-1}-\eta^2\left[\delta(1-P_3)\right]^{n-2}- (1-S)\eta^2P_3(n-m-1)\left[\delta(1-P_3)\right]^{n-2}  \notag \\
  {}& +\widetilde{S}\eta P_3 \left[(\delta+\eta)(1-P_3)\right]^{n-2} m \left(\frac{\delta}{\delta+\eta}\right)^{m-1} - (1-S)\delta\eta P_3^2 m[\delta(1-P_3)]^{n-2}
\end{align}

The probability $\pi_{(n,m,l)}$, $n>m>l\geq1$ is solved as
\begin{align}
\pi_{(n,m,l)}&= \eta^2(1-P_3)\bigg\{ \left(pP_1+\eta\right)\delta^{n-3}\left(\frac{(\delta+\eta)(1-P_3)}{\delta}\right)^{m-2} \left(\frac{\delta}{\delta+\eta}\right)^{l-1}   \notag \\
& \qquad -\eta[\delta(1-P_3)]^{n-3}\left(\frac{\delta+\eta}{\delta}\right)^{m-2} \left(\frac{\delta}{\delta+\eta}\right)^{l-1} \notag \\
& \qquad -(1-S)\eta P_3(n-m-1)    [\delta(1-P_3)]^{n-3}\left(\frac{\delta+\eta}{\delta}\right)^{m-2}\left(\frac{\delta}{\delta+\eta}\right)^{l-1} \bigg\}  \notag \\
& \quad + \eta P_3\widetilde{S}[(\delta+\eta)(1-P_3)]^{n-2}\left(\frac{\delta}{\delta+\eta}\right)^{l-1} - \eta P_3\widetilde{S}[(\delta+\eta)(1-P_3)]^{n-2}\left(\frac{\delta}{\delta+\eta}\right)^{m-1}  \notag \\
& \quad -\delta\eta P_3^2(1-S)[\delta(1-P_3)]^{n-2}\left(\frac{\delta+\eta}{\delta}\right)^{m-1} \left(\frac{\delta}{\delta+\eta}\right)^{l-1}  \notag \\
&\quad + \delta\eta P_3^2(1-S)[\delta(1-P_3)]^{n-2}
\end{align}
where the numbers $S$ and $\widetilde{S}$ are determined by
\begin{align}
\widetilde{S}&=S\eta+(1-S)(\delta+\eta)P_3,   \\
S&=\frac{pP_1[1-\delta(1-P_3)]^2+\delta\eta P_3^2}{(1-\delta)[1-\delta(1-P_3)]^2-\delta\eta P_3(1-\delta)(1-P_3)}
\end{align}
\end{Theorem}

We will solve the stationary equations (14) in Appendix D.

Provided all the stationary probabilities, the distribution of steady-state AoI can be obtained by equation (15). The calculation details are given in Appendix E.

\begin{Theorem}
Assume that the relay has a buffer of size 1, for $n\geq1$, the steady-state AoI $\Delta_B^P$ of the relay-assisted status updating system is distributed as
\begin{align}
\Pr\{\Delta_B^P=n\}&=\frac{(\delta+\eta)P_3(pP_1+\eta)}{\delta-(\delta+\eta)(1-P_3)}\delta^{n-1}-c_1[(\delta+\eta)(1-P_3)]^{n-1} +c_2[\delta(1-P_3)]^{n-1} \notag \\
& \quad  +\frac{(1-S)(\delta+\eta)P_3^2}{1-P_3}n[\delta(1-P_3)]^{n-1} + \frac{P_3\widetilde{S}}{1-P_3}n[(\delta+\eta)(1-P_3)]^{n-1}
\end{align}
where the coefficients $c_1$ and $c_2$ are given as
\begin{align}
c_1&=\frac{pP_1\eta(1-P_3)+\delta\eta P_3}{[\delta-(\delta+\eta)(1-P_3)](1-P_3)} + \frac{\delta P_3[\eta+(\delta+\eta)P_3](1-S)+(\delta+\eta)P_3\widetilde{S}}{\eta(1-P_3)}   \\
c_2&=\frac{P_3[\eta(\delta+\eta)+P_3(1-S)(\delta+\eta)(2\delta-\eta)]}{\eta(1-P_3)}
\end{align}
\end{Theorem}

\section{Performance metrics of Age of Information}
In Section III, for two cases where the packet retransmission from relay can and cannot be preempted by later packets, we derive the stationary AoI distribution of status updating system with a no-buffer relay. To investigate the influence of \emph{storing another packet} in relay on system's AoI, the distribution of steady state AoI is determined in Section IV.

Given the stationary AoI distributions, the following corollary shows the mean and the variance of AoI for all three cases, and the calculation details are provided in Appendix F.
\begin{Corollary}
For the relay-assisted status updating system, the average AoI for no-buffer-and-preemption case is equal to
\begin{align}
\mathbb{E}[\Delta_{NB}^P]&=\frac{(1-pP_1)P_3}{[(1-pP_1)P_3-p(1-P_1)P_2](1-\delta)} + \frac{\xi}{[1-(1-pP_1)(1-P_3)]^2}
\end{align}
while for the no-buffer-no-preemption case, it shows that
\begin{align}
\mathbb{E}[\Delta_{NB}^{NP}]&= \frac{\beta_1}{(1-\delta)^2}+\frac{\beta_2}{[1-(1-pP_1)(1-P_3)]^2}  \notag \\
& \quad  + \frac{p(1-P_1)P_2\beta_1[2-\delta-(1-pP_1)(1-P_3)]}{(1-\delta)^2[1-(1-pP_1)(1-P_3)]^2} +\frac{2p(1-P_1)P_2\beta_2}{[1-(1-pP_1)(1-P_3)]^3}
\end{align}

The average AoI of relay-assisted system when the relay is equipped with a buffer of size 1 is calculated as
\begin{align}
\mathbb{E}[\Delta_B^P]&=\frac{(\delta+\eta)P_3(pP_1+\eta)}{[\delta-(\delta+\eta)(1-P_3)](1-\delta)^2}-\frac{c_1}{[1-(\delta+\eta)(1-P_3)]^2} +\frac{c_2}{[1-\delta(1-P_3)]^2}  \notag \\
& \quad +\frac{(1-S)(\delta+\eta)P_3^2[1+\delta(1-P_3)]}{(1-P_3)[1-\delta(1-P_3)]^3} +\frac{P_3\widetilde{S}[1+(\delta+\eta)(1-P_3)]}{(1-P_3)[1-(\delta+\eta)(1-P_3)]^3}
\end{align}
where in all of three equations
\begin{equation}
\delta=1-p(P_1+P_2-P_1P_2)  \notag
\end{equation}

Besides, $\xi$ in (24) is given as
\begin{align}
\xi&=\frac{p[P_1+(1-P_1)P_2P_3]}{1-P_3} - \frac{P_3(1-\delta)\delta}{[(1-pP_1)P_3-p(1-P_1)P_2](1-P_3)}  \notag
\end{align}
and $\beta_1$ and $\beta_2$ in expression (25) are defined in (9) and (10).

In (26), $\eta=p(1-P_1)P_2$, $S$, $\widetilde{S}$ are given in (19) and (20). Other two coefficients $c_1$, $c_2$ are introduced in equations (22) and (23).
\end{Corollary}

For different situations, the second moments of AoI can also be obtained. Then, we can calculate the variance of stationary AoI by
\begin{equation}
\emph{Var}[\Delta]=\mathbb{E}[\Delta^2]-\left(\mathbb{E}[\Delta]\right)^2
\end{equation}
for all three cases.

The expression of $\emph{Var}[\Delta]$ is lengthy but the calculations are direct. Here, we provide only $\emph{Var}[\Delta_{NB}^{NP}]$, AoI variances for other two cases can be obtained similarly.

From AoI distribution (11), we have
\begin{align}
\mathbb{E}\left[\left(\Delta_{NB}^{NP}\right)^2\right] &= \left(\beta_1+\frac{p(1-P_1)P_2\beta_1}{\delta-(1-pP_1)(1-P_3)}\right)\sum\nolimits_{n=1}^{\infty}n^2\delta^{n-1}    \notag \\
& \quad + \left(\beta_2 - \frac{p(1-P_1)P_2\beta_1}{\delta-(1-pP_1)(1-P_3)}\right) \sum\nolimits_{n=1}^{\infty}n^2[(1-pP_1)(1-P_3)]^{n-1}    \notag \\
& \quad + p(1-P_1)P_2\beta_2\sum\nolimits_{n=1}^{\infty}n^2(n-1)[(1-pP_1)(1-P_3)]^{n-2}  \notag \\
&= \left(\beta_1+\frac{p(1-P_1)P_2\beta_1}{\delta-(1-pP_1)(1-P_3)}\right)\frac{1+\delta}{(1-\delta)^3} +\left(\beta_2 - \frac{p(1-P_1)P_2\beta_1}{\delta-(1-pP_1)(1-P_3)}\right)   \notag \\
& \quad  \times \frac{1+(1-pP_1)(1-P_3)}{[1-(1-pP_1)(1-P_3)]^3} + p(1-P_1)P_2\beta_2 \frac{4+2(1-pP_1)(1-P_3)}{[1-(1-pP_1)(1-P_3)]^4}
\end{align}
where the last sum
\begin{align}
&\sum\nolimits_{n=1}^{\infty}n^2(n-1)[(1-pP_1)(1-P_3)]^{n-2}   \notag \\
={}& \sum\nolimits_{n=1}^{\infty}(n+1)n(n-1)[(1-pP_1)(1-P_3)]^{n-2} -\sum\nolimits_{n=1}^{\infty}n(n-1)[(1-pP_1)(1-P_3)]^{n-2}   \notag \\
={}&\frac{6}{[1-(1-pP_1)(1-P_3)]^4} - \frac{2}{[1-(1-pP_1)(1-P_3)]^3} \frac{4+2(1-pP_1)(1-P_3)}{[1-(1-pP_1)(1-P_3)]^4}  \notag
\end{align}
can be derived using following formula
\begin{equation}
\sum\nolimits_{n=1}^{\infty}n(n-1)\cdots(n-k+1)x^{n-k}=\frac{k!}{(1-x)^{k+1}} \text{  for }0<x<1, \text{ and any }k\geq1.    \notag
\end{equation}

Combining results (25) and (28), the AoI variance $\emph{Var}\left[\Delta_{NB}^{NP}\right]$ is obtained by formula (27).

\section{Numerical Simulations}
The numerical results of AoI are provided in this Section. Specifically, we will depict the stationary distribution of three AoIs, i.e., $\Delta_{NB}^P$, $\Delta_{NB}^{NP}$ and $\Delta_B^P$ for relay-assisted status updating system. For three situations, we also draw the curves of average AoI and the AoI's variance. To demonstrate the improvement of system's timeliness by adding an intermediate relay, the
AoI performance of a non-relay updating system is calculated and comparisons are carried out as well.

First of all, the age of information of a system without relay is analyzed briefly. When there is no relay in system, the packet can only be transmitted to destination via the direct link $s-d$. Since the packet on channel $s-d$ is assumed to be erased independently at each time slot, and the propagation delay is 1, it is easy to describe the random dynamics of the AoI at destination $d$. Let the current AoI is $n$, it is observed that either the value of AoI is reduced to 1 when a packet is obtained successfully on $s-d$ link, which occurs with probability $pP_1$, or the AoI increases 1. These random state transfers can be characterized by a one-dimensional stochastic process whose discrete state is simply defined as the AoI at $d$.

When the process reaches the steady state, it is not hard to obtain that the AoI of non-relay system, $\Delta$, is distributed as
\begin{equation}
\Pr\{\Delta=n\}=pP_1(1-pP_1)^{n-1}  \qquad  (n\geq1)
\end{equation}
which is a geometric distribution with parameter $pP_1$.

The mean and variance of non-relay AoI are calculated as
\begin{equation}
\mathbb{E}[\Delta]=\frac{1}{pP_1}, \qquad  \quad
\emph{Var}[\Delta]=\frac{1-pP_1}{(pP_1)^2}
\end{equation}

\begin{figure}[!t]
\centering
\includegraphics[width=2.4in]{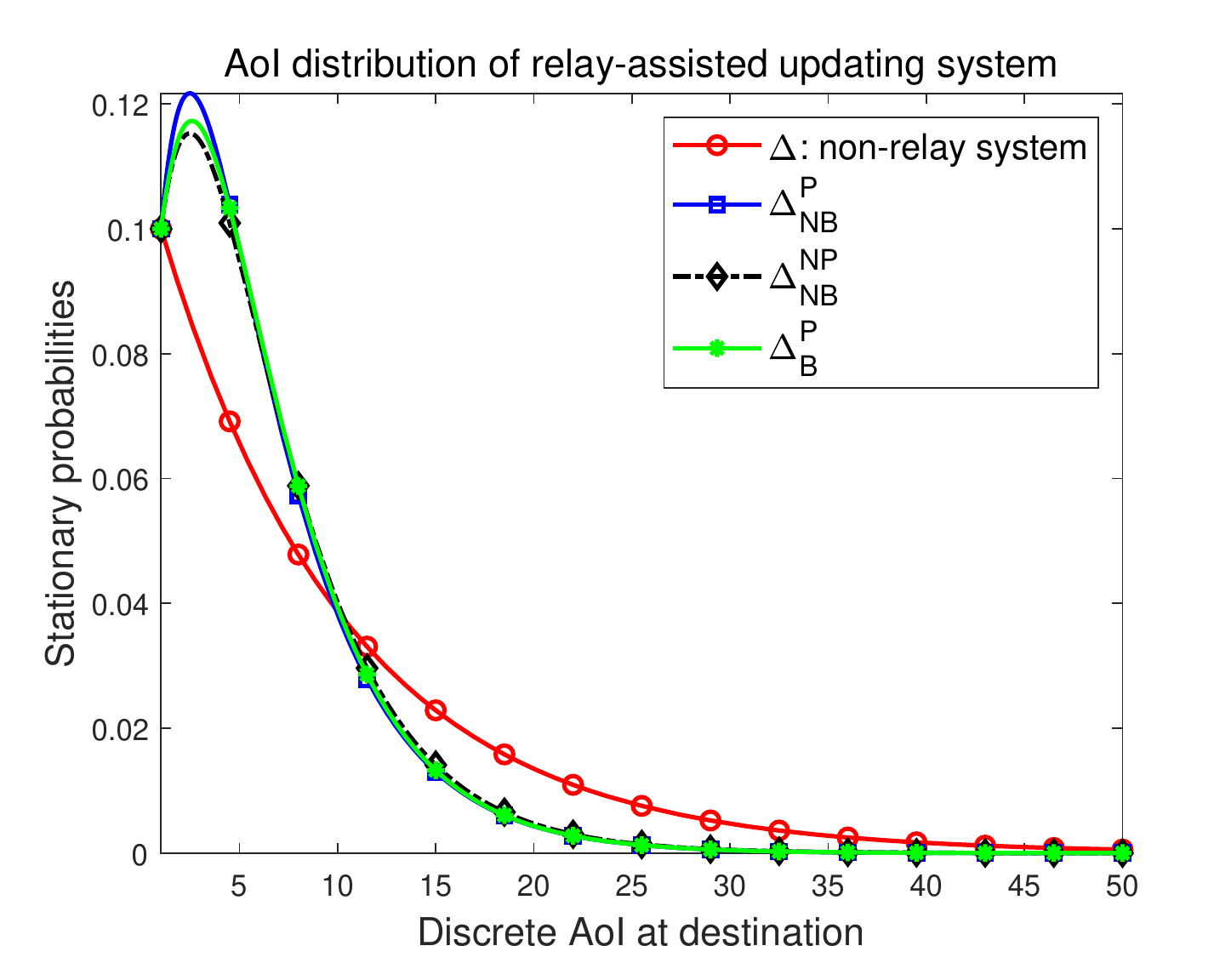}
\caption{Stationary AoI distributions of relay-assisted status updating system.}
\label{fig_dis}
\end{figure}

We depict the AoI's distribution in Figure \ref{fig_dis}, including three relay-assisted cases and the AoI distribution of a non-relay system. The packet generation probability is set as $p=2/5$. The transmission success probabilities of three links are $P_1=1/4$, $P_2=1/3$, and $P_3=1/3$.

It was seen that when the relay is added, all of three distribution curves have a peak, and decrease rapidly as the value of AoI becomes larger. While the AoI distribution of the non-relay system is decreasing in all the AoI range, whose tail drops much slower. Therefore, compared with the non-relay system, it is less likely that the AoI of a relay-assisted status updating system takes large values. The difference of three relay-assisted AoI distributions are tiny. It is seen that three distribution curves take the peak value under approximately the same AoI, and almost coincide when the AoI get slightly large, for intance, when its value is greater than 5.

\begin{figure*}[!t]
\centering
\subfloat{\includegraphics[width=2.2in]{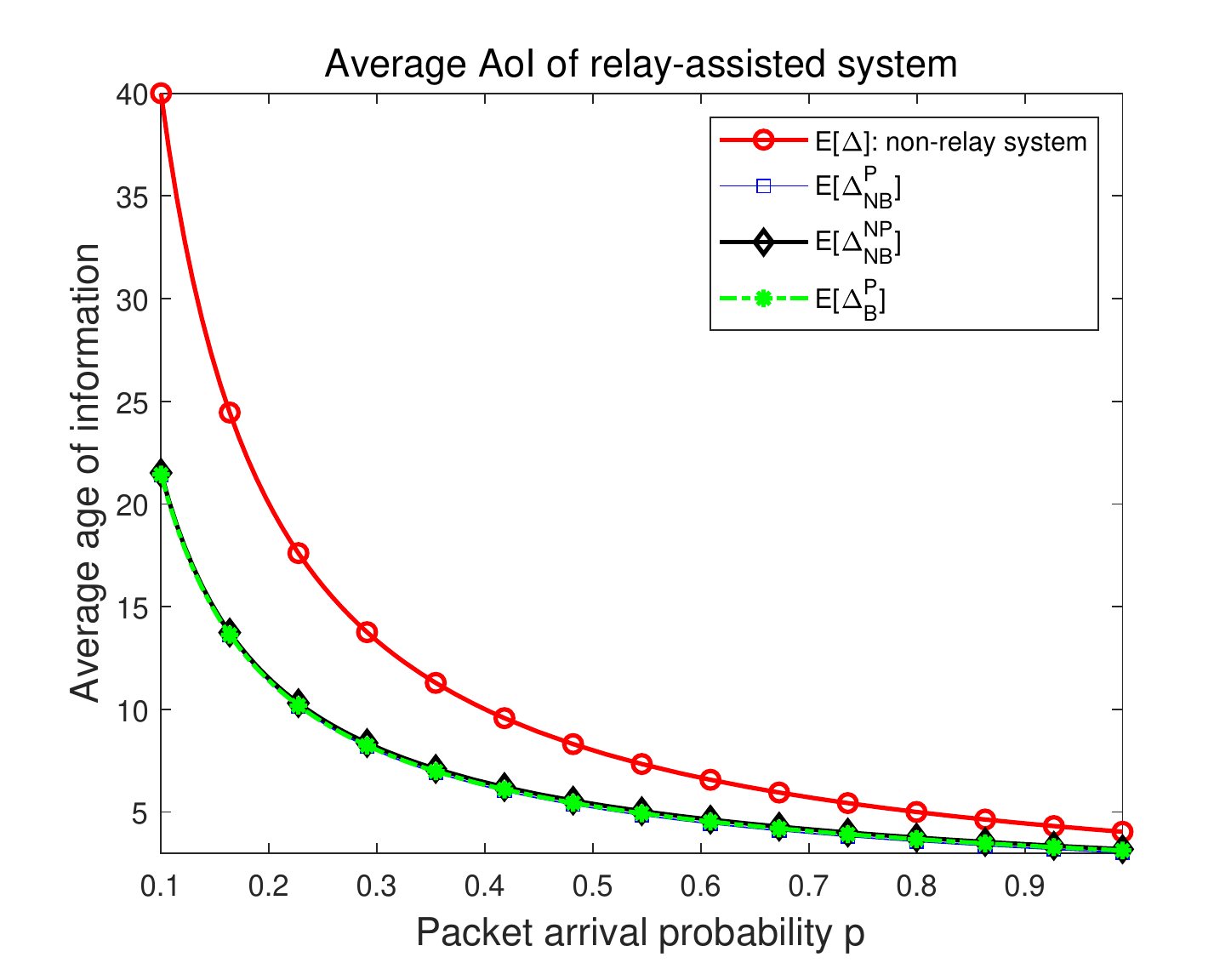}%
\label{fig_ave_1}}
\hfil
\subfloat{\includegraphics[width=2.2in]{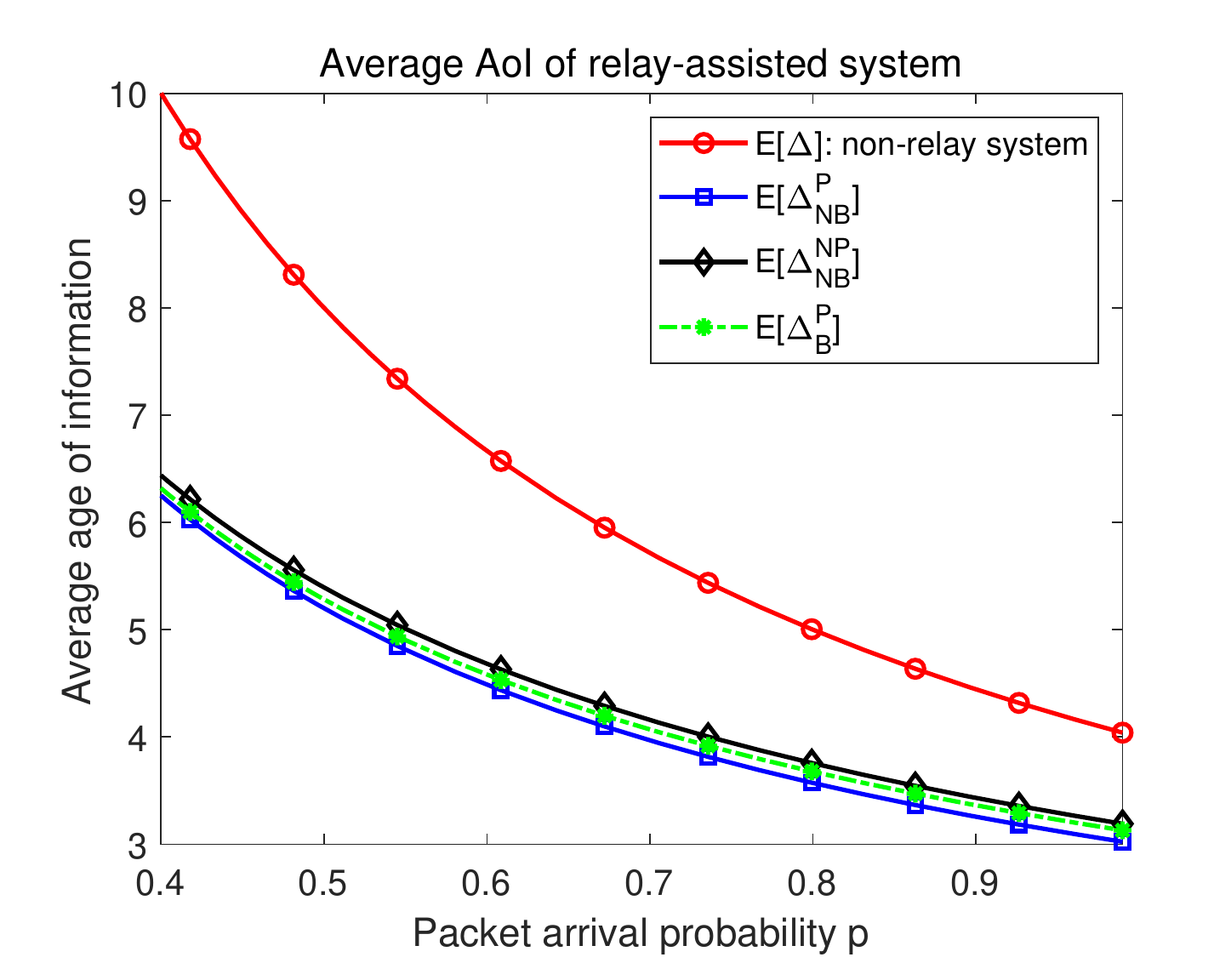}%
\label{fig_ave_2}}
\hfil
\caption{Average AoI of updating system with and without relay.}
\label{fig_ave}
\end{figure*}

To compare the timeliness performance of various updating systems, we depict the graphs of average AoI in Figure \ref{fig_ave}. The relationship among three average AoI of relay-assisted situations are clear in Figure \ref{fig_ave_2}, in which we offer the numerical results of all the cases when the packet generation probability is high. Firstly, it is apparent that introducing a relay can reduce the average AoI at the destination. The average AoI of relay-assisted cases are smaller than that of a non-relay system. As packet generation probability $p$ gets larger, the average performance of AoI are all decreasing, no matter there has a relay in system or not. At the same time, the average AoI's gap between relay-assisted cases and non-relay case is getting small.

Numerical results in Figure \ref{fig_ave_2} shows that the average AoI $\mathbb{E}\left[\Delta_{NB}^P\right]$ is minimal among three relay-assisted situations. It is understandable that average age of information $\mathbb{E}\left[\Delta_{NB}^P\right]$ is less than $\mathbb{E}\left[\Delta_{NB}^{NP}\right]$, since replacing an “old” packet with a fresh one does not increase the AoI at the destination. However, observing that $\mathbb{E}\left[\Delta_{B}^P\right]$ is above $\mathbb{E}\left[\Delta_{NB}^P\right]$, which implies that under the system settings in this paper, storing another packet in relay does not help to reduce the system's AoI. Equivalently, from the sense of average performance, adding buffer in relay cannot enhance the timeliness of a relay-assisted status updating system. Mathematically, we have
\begin{equation}
   \mathbb{E}\left[\Delta_{NB}^P\right]
< \mathbb{E}\left[\Delta_B^P\right]
< \mathbb{E}\left[\Delta_{NB}^{NP}\right]
< \mathbb{E}\left[\Delta\right]
\end{equation}

\begin{figure*}[!t]
\centering
\subfloat{\includegraphics[width=2.2in]{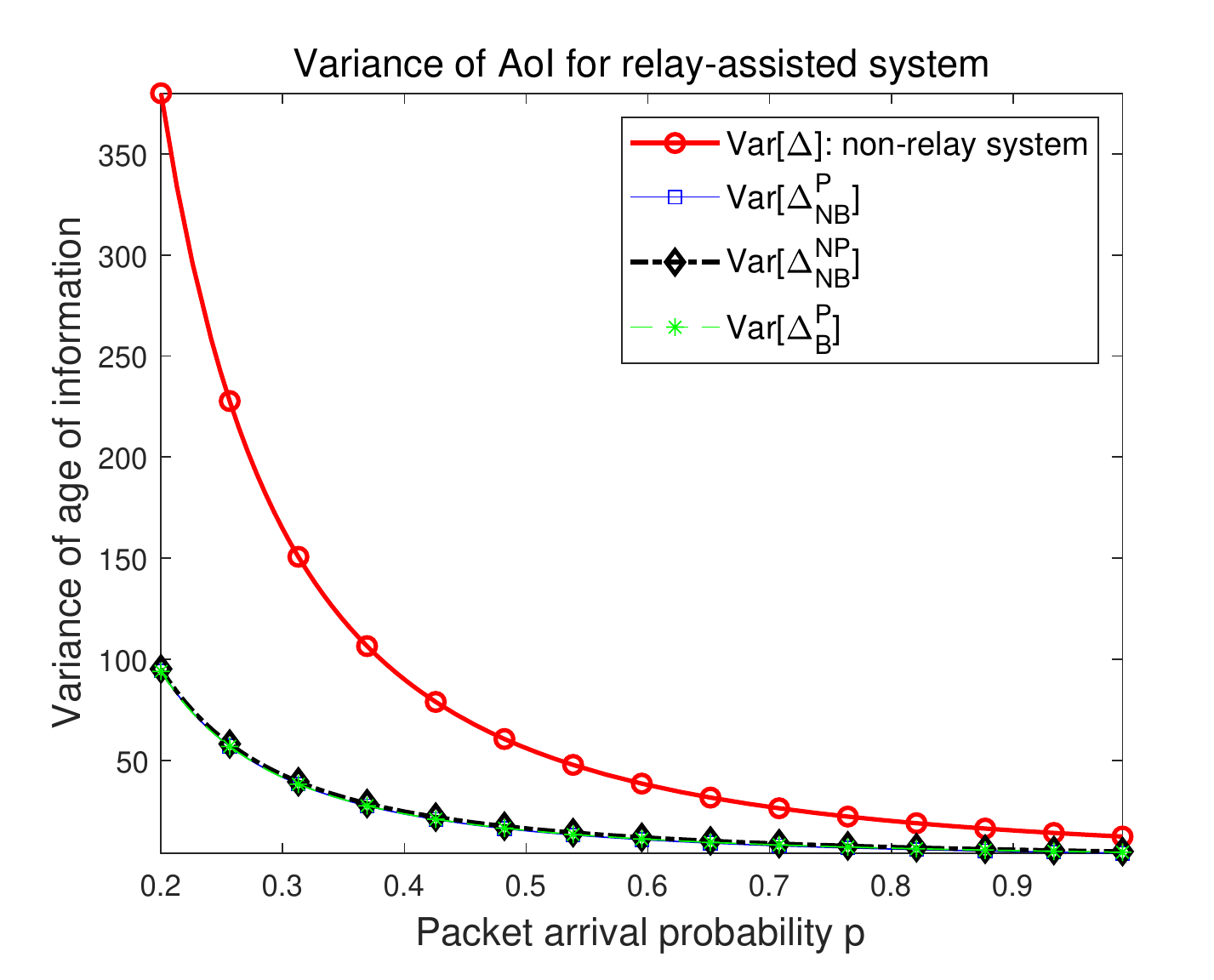}%
\label{fig_var_1}}
\hfil
\subfloat{\includegraphics[width=2.2in]{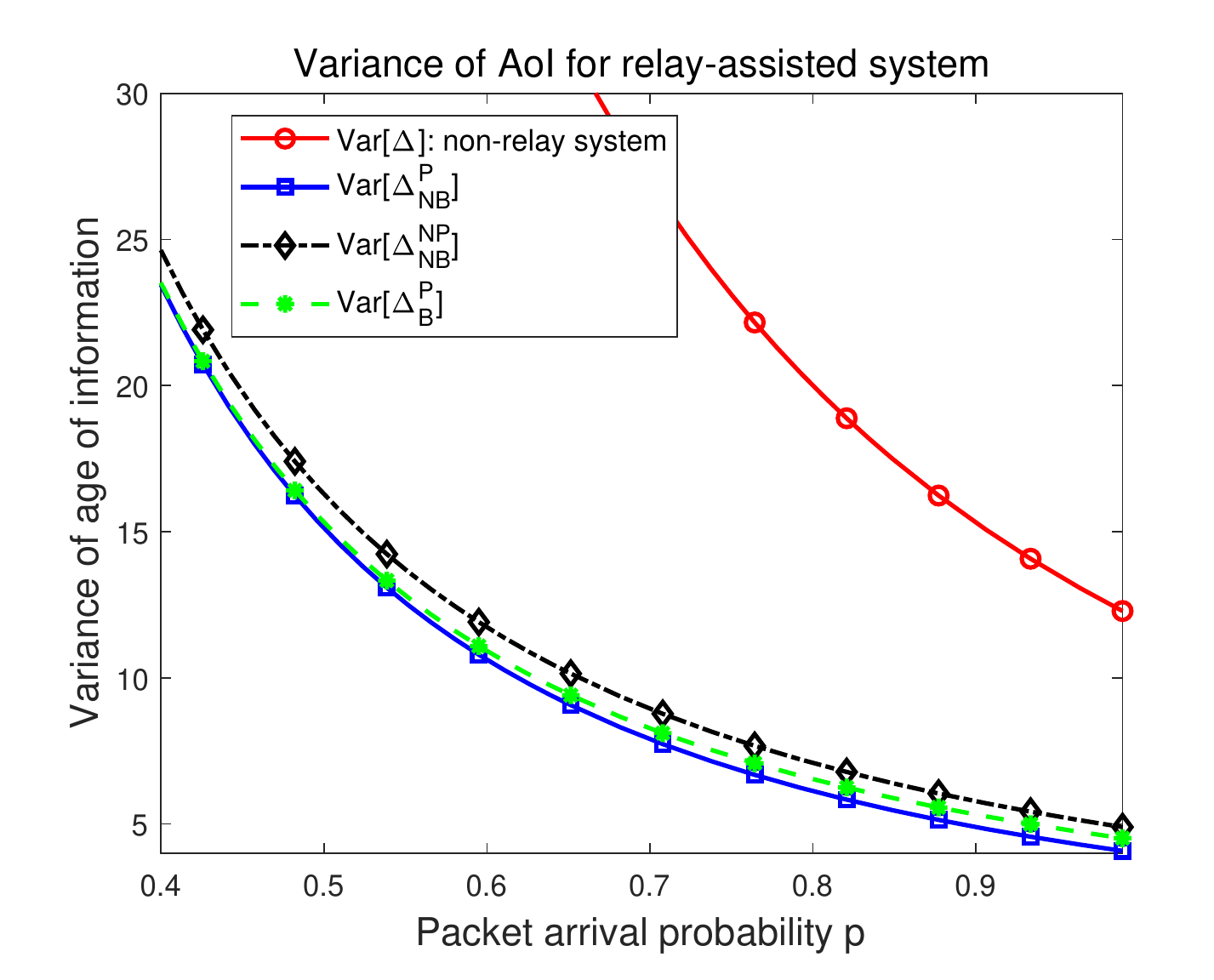}%
\label{fig_var_2}}
\hfil
\caption{Variance of AoI for the status updating system with and without relay.}
\label{fig_var}
\end{figure*}

In Figure \ref{fig_var}, we also provide the variance of stationary AoI for each case. Simulation results show that
\begin{equation}
   \emph{Var}\left[\Delta_{NB}^P\right]
< \emph{Var}\left[\Delta_B^P\right]
< \emph{Var}\left[\Delta_{NB}^{NP}\right]
< \emph{Var}\left[\Delta\right]
\end{equation}

Notice that the relationships between various AoI variances are the same as that of average AoIs. The age of information $\Delta_{NB}^P$ also has the minimal variance. The comparisons between three AoI variances for relay-assisted situations are illustrated in Figure \ref{fig_var_2}. When the packet generation probability is large, the relationships among them is clearer.

\section{Conclusion}
In this paper, we considere a relay-assisted status updating system and analyze the performace of age of information at the destination. We characterize the stationary AoI for three different settings at relay, i.e., no-buffer-and-preemption, no-buffer-no-preemption, and buffer-and-preemption. The obtained results show that introducing an intermediate relay can reduce the average AoI and AoI's variance dramatically, thus proving the improvement of system's timeliness by adding the relay. On the other hand, we show that the no-buffer-and-preemption setting of relay achieve the minimal average performance and variance of the AoI. As a result, under the system model considered in this paper, there is no necessity to add the buffer in relay because this does not help to enhance the system's timeliness.

We attribute the idea that invoking a multi-dimensional state vector and constituting multi-dimensional AoI stochastic process to SHS method, which is a systematic approach to calculate the average AoI of status updating systems in continuous time model. The methods used in the paper are the standard tools from queueing theory. Notice that in a status updating system, it is not easy to describe the AoI itself, but indeed is easier to characterize the dynamics of a “bigger” state vector which contains the AoI as the part. Therefore, as long as the steady state of the larger-dimensional process is solved, as one of marginal distributions, the stationary distribution of AoI is obtained as well. Once all the stationary probabilities of larger-dimensional process are obtained, apart from AoI, other marginal distributions are all be determined. In this paper, for example, we can calculate the age distribution of packet in relay by summing over AoI-component in state vector.


\bibliographystyle{IEEEtranTCOM}
\bibliography{paperlist}

\begin{thebibliography}{10}
\baselineskip 12pt
\providecommand{\url}[1]{#1}
\csname url@samestyle\endcsname
\providecommand{\newblock}{\relax}
\providecommand{\bibinfo}[2]{#2}
\providecommand{\BIBentrySTDinterwordspacing}{\spaceskip=0pt\relax}
\providecommand{\BIBentryALTinterwordstretchfactor}{4}
\providecommand{\BIBentryALTinterwordspacing}{\spaceskip=\fontdimen2\font plus
\BIBentryALTinterwordstretchfactor\fontdimen3\font minus
  \fontdimen4\font\relax}
\providecommand{\BIBforeignlanguage}[2]{{%
\expandafter\ifx\csname l@#1\endcsname\relax
\typeout{** WARNING: IEEEtran.bst: No hyphenation pattern has been}%
\typeout{** loaded for the language `#1'. Using the pattern for}%
\typeout{** the default language instead.}%
\else
\language=\csname l@#1\endcsname
\fi
#2}}
\providecommand{\BIBdecl}{\relax}
\BIBdecl

\bibitem{1}
S.~Kaul, M.~Gruteser, V.~Rai, and J.~Kenney, ``Minimizing age of information in
  vehicular networks,'' in \emph{8th Annu. IEEE Commun. Soc. Conf. Sensor, Mesh
  Ad Hoc Commun. Netw. (SECON)}, Jun. 2011, pp. 350--358.

\bibitem{2}
R.~D. Yates, Y.~Sun, D.~R.~B. III, S.~K. Kaul, E.~Modiano, and S.~Ulukus, ``Age
  of information: An introduction and survey,'' \emph{IEEE J. Sel. Areas
  Commun.}, pp. 1--1, 2021.

\bibitem{3}
S.~Kaul, R.~Yates, and M.~Gruteser, ``Real-time status: How often should one
  update?'' in \emph{IEEE INFOCOM}, Mar. 2012, pp. 2731--2735.

\bibitem{4}
------, ``Status updates through queues,'' in \emph{46th Annu. Conf. Inf. Sci.
  Syst. (CISS)}, Mar. 2012, pp. 1--6.

\bibitem{5}
R.~D. Yates and S.~Kaul, ``The age of information: Real-time status updating by
  multiple sources,'' \emph{IEEE Trans. Inf. Theory}, vol.~65, no.~3, pp.
  1807--1827, 2019.

\bibitem{6}
M.~Moltafet, M.~Leinonen, and M.~Codreanu, ``Moment generating function of the
  aoi in a two-source system with packet management,'' \emph{IEEE Wirel.
  Commun. Lett.}, pp. 1--1, 2020.

\bibitem{7}
O.~Dogan and N.~Akar, ``The multi-source preemptive m/ph/1/1 queue with packet
  errors: Exact distribution of the age of information and its peak,''
  \emph{arXiv:2007.11656v1}, 2020.

\bibitem{8}
S.~K. Kaul and R.~D. Yates, ``Timely updates by multiple sources: The m/m/1
  queue revisited,'' in \emph{54th Annu. Conf. Inf. Sci. Syst. (CISS)}, Mar.
  2020, pp. 1--6.

\bibitem{9}
R.~D. Yates and S.~K. Kaul, ``Age of information in uncoordinated unslotted
  updating,'' in \emph{IEEE Int. Symp. Inf. Theory (ISIT)}, 2020, pp.
  1759--1764.

\bibitem{10}
M.~Moltafet, M.~Leinonen, and M.~Codreanu, ``Average aoi in multi-source
  systems with source-aware packet management,'' \emph{IEEE Trans. Commun.},
  vol.~69, no.~2, pp. 1121--1133, 2021.

\bibitem{11}
A.~Maatouk, M.~Assaad, and A.~Ephremides, ``Minimizing the age of information:
  Noma or oma?'' in \emph{IEEE INFOCOM Workshops)}, Apr. 2019, pp. 102--108.

\bibitem{12}
R.~D. Yates, ``The age of information in networks: Moments, distributions, and
  sampling,'' \emph{IEEE Trans. Inf. Theory}, vol.~66, no.~9, pp. 5712--5728,
  2020.

\bibitem{13}
M.~Costa, M.~Codreanu, and A.~Ephremides, ``Age of information with packet
  management,'' in \emph{IEEE Int. Symp. Inf. Theory (ISIT)}, 2014, pp.
  1583--1587.

\bibitem{14}
C.~Kam, S.~Kompella, G.~D. Nguyen, J.~E. Wieselthier, and A.~Ephremides, ``On
  the age of information with packet deadlines,'' \emph{IEEE Trans. Inf.
  Theory}, vol.~64, no.~9, pp. 6419--6428, 2018.

\bibitem{15}
M.~A. Abd-Elmagid, H.~S. Dhillon, and N.~Pappas, ``A reinforcement learning
  framework for optimizing age of information in rf-powered communication
  systems,'' \emph{IEEE Trans. Commun.}, vol.~68, no.~8, pp. 4747--4760, 2020.

\bibitem{16}
S.~Kriouile and M.~Assaad, ``Minimizing the age of incorrect information for
  real-time tracking of markov remote sources,'' \emph{arXiv:2102.03245}, 2021.

\bibitem{17}
V.~Tripathi and E.~Modiano, ``Age debt: A general framework for minimizing age
  of information,'' \emph{arXiv:2101.10225}, 2021.

\bibitem{18}
H.~Tang, P.~Ciblat, J.~Wang, M.~Wigger, and R.~D. Yates, ``Cache updating
  strategy minimizing the age of information with time-varying files'
  popularities,'' \emph{arXiv:2010.04787}, 2021.

\bibitem{19}
M.~Moltafet, M.~Leinonen, M.~Codreanu, and N.~Pappas, ``Power minimization for
  age of information constrained dynamic control in wireless sensor networks,''
  \emph{arXiv:2007.05364}, 2021.

\bibitem{20}
A.~Maatouk, S.~Kriouile, M.~Assad, and A.~Ephremides, ``On the optimality of
  the whittle's index policy for minimizing the age of information,''
  \emph{IEEE Trans. Wirel. Commun.}, vol.~20, no.~2, pp. 1263--1277, 2021.

\bibitem{21}
H.~Tang, J.~Wang, L.~Song, and J.~Song, ``Minimizing age of information with
  power constraints: Multi-user opportunistic scheduling in multi-state
  time-varying channels,'' \emph{IEEE J. Sel. Areas Commun.}, vol.~38, no.~5,
  pp. 854--868, 2020.

\bibitem{22}
H.~Tang, J.~Wang, Z.~Tang, and J.~Song, ``Scheduling to minimize age of
  synchronization in wireless broadcast networks with random updates,''
  \emph{IEEE Trans. Wirel. Commun.}, vol.~19, no.~6, pp. 4023--4037, 2020.

\bibitem{23}
C.~Kam, S.~Kompella, G.~D. Nguyen, and A.~Ephremides, ``Effect of message
  transmission path diversity on status age,'' \emph{IEEE Trans. Inf. Theory},
  vol.~62, no.~3, pp. 1360--1374, 2016.

\bibitem{24}
O.~Ayan, H.~M. Gürsu, A.~Papa, and W.~Kellerer, ``Probability analysis of age
  of information in multi-hop networks,'' \emph{IEEE Netw. Lett.}, vol.~2,
  no.~2, pp. 76--80, 2020.

\bibitem{25}
B.~Li, Q.~Wang, H.~Chen, Y.~Zhou, and Y.~Li, ``Optimizing information freshness
  for cooperative iot systems with stochastic arrivals,'' \emph{IEEE Internet
  Things J.}, pp. 1--1, 2021.

\bibitem{26}
G.~Kesidis, T.~Konstantopoulos, and M.~Zazanis, ``Age of information
  distribution under dynamic service preemption,'' \emph{arXiv:2104.11393},
  2021.

\bibitem{27}
------, ``Age of information without service preemption,''
  \emph{arXiv:2104.08050}, 2021.

\bibitem{28}
------, ``The new age of information: a tool for evaluating the freshness of
  information in bufferless processing systems,'' \emph{Queueing Syst.},
  no.~95, pp. 203--250, 2020.

\bibitem{29}
------, ``Age of information for small buffer systems,''
  \emph{arXiv:2106.08473}, 2021.

\bibitem{30}
N.~Akar, O.~Doğan, and E.~U. Atay, ``Finding the exact distribution of (peak)
  age of information for queues of ph/ph/1/1 and m/ph/1/2 type,'' \emph{IEEE
  Trans. Commun.}, vol.~68, no.~9, pp. 5661--5672, 2020.

\bibitem{31}
N.~Akar, ``Discrete-time queueing model of age of information with multiple
  information sources,'' \emph{IEEE Internet Things J.}, pp. 1--1, 2021.

\bibitem{32}
O.~Dogan and N.~Akar, ``The multi-source preemptive m/ph/1/1 queue with packet
  errors: Exact distribution of the age of information and its peak,''
  \emph{arXiv:2007.11656v1}, 2020.

\bibitem{33}
Y.~Inoue, H.~Masuyama, T.~Takine, and T.~Tanaka, ``A general formula for the
  stationary distribution of the age of information and its application to
  single-server queues,'' \emph{IEEE Trans. Inf. Theory}, vol.~65, no.~12, pp.
  8305--8324, 2019.

\end{thebibliography}

\appendices


\section{Proof of Theorem 2}
In this Appendix, we solve the following system of equations and find all the probabilities $\pi_{(n,m)}$, $n>m\geq0$.
\begin{equation}
\begin{cases}
\pi_{(n,m)}=\pi_{(n-1,m-1)}\delta_1                & \qquad   (n>m\geq2)  \\
\pi_{(n,1)}=\pi_{(n-1,0)}p(1-P_1)P_2 + \sum\nolimits_{j=1}^{n-2}\pi_{(n-1,j)}p(1-P_1)P_2(1-P_3)   \\
\qquad \qquad + \sum\nolimits_{k=n}^{\infty}\pi_{(k,n-1)}p(1-P_1)P_2P_3      & \qquad    (n\geq3)  \\
\pi_{(2,1)}=\pi_{(1,0)}p(1-P_1)P_2 + \sum\nolimits_{k=2}^{\infty}\pi_{(k,1)}p(1-P_1)P_2P_3   \\
\pi_{(n,0)}=\pi_{(n-1,0)}\delta_2 + \sum\nolimits_{k=n}^{\infty}\pi_{(k,n-1)}\delta_3    & \qquad   (n\geq2)  \\
\pi_{(1,0)}=\left(\sum\nolimits_{n=1}^{\infty}\pi_{(n,0)} + \sum\nolimits_{m=1}^{\infty}\sum\nolimits_{n=m+1}^{\infty}\pi_{(n,m)}   \right) pP_1
\end{cases}
\end{equation}
where we denote
\begin{align}
\delta_2&=(1-p) + p(1-P_1)(1-P_2) = 1-p(P_1+P_2-P_1P_2)    \notag
\end{align}
and
\begin{align}
\delta_1&=(1-p)(1-P_3)+p(1-P_1)(1-P_2)(1-P_3) =(1-P_3)\delta_2,  \notag \\
\delta_3&=(1-p)P_3+p(1-P_1)(1-P_2)P_3=P_3\delta_2  \notag
\end{align}

First of all, from the first line of (33) we have
\begin{align}
\pi_{(n,m)}&= \pi_{(n-1,m-1)}\delta_1=\dots=\pi_{(n-m+1,1)}\delta_1^{m-1}   \notag \\
&= \begin{cases}
\left[\pi_{(1,0)}p(1-P_1)P_2 \sum\nolimits_{k=2}^{\infty}\pi_{(k,1)}p(1-P_1)P_2P_3 \right]\delta_1^{m-1}   &  (n-m=1)  \\
\Big[\pi_{(n-m,0)}p(1-P_1)P_2  \sum\nolimits_{j=1}^{n-m-1}\pi_{(n-m,j)}p(1-P_1)P_2(1-P_3)  \\
  \qquad \quad + \sum\nolimits_{k=n-m+1}^{\infty}\pi_{(k,n-m)}p(1-P_1)P_2P_3 \Big]\delta_1^{m-1}  &    (n-m\geq2)
\end{cases}
\end{align}

In above expressions, we have substituted the second and third lines of (33).

For cases $n\geq2$, the fourth equation of (33) says that
\begin{align}
\pi_{(n,0)}=\pi_{(n-1,0)}\delta_2 + \left(\sum\nolimits_{k=n}^{\infty}\pi_{(k,n-1)}\right) \delta_3
\end{align}

Using probability expressions obtained in (34), we deal with the sum in equation (35) as follows.
\begin{align}
&\sum\nolimits_{k=n}^{\infty}\pi_{(k,n-1)}= \pi_{(n,n-1)}+\sum\nolimits_{k=n+1}^{\infty}\pi_{(k,n-1)} \notag \\
={}&\left[\pi_{(1,0)}p(1-P_1)P_2 + \sum\nolimits_{k=2}^{\infty}\pi_{(k,1)}p(1-P_1)P_2P_3 \right]\delta_1^{n-2}  \notag \\
  {}& \quad + \sum\nolimits_{k=n+1}^{\infty} \bigg\{ \pi_{(k-n+1,0)}p(1-P_1)P_2 + \sum\nolimits_{j=1}^{k-n}\pi_{(k-n+1,j)}p(1-P_1)P_2(1-P_3)   \notag \\
  {}& \quad + \sum\nolimits_{y=k-n+2}^{\infty}\pi_{(y,k-n+1)}p(1-P_1)P_2P_3 \bigg\}\delta_1^{n-2}  \notag  \\
={}&\left[\pi_{(1,0)}p(1-P_1)P_2 + \left(\sum\nolimits_{k=2}^{\infty}\pi_{(k,1)}\right) p(1-P_1)P_2P_3 \right]\delta_1^{n-2} \notag \\
  {}& \quad + \bigg[ \left(\sum\nolimits_{n=2}^{\infty}\pi_{(n,0)}\right) p(1-P_1)P_2 + \left( \sum\nolimits_{k=n+1}^{\infty} \sum\nolimits_{j=1}^{k-n}\pi_{(k-n+1,j)}\right) p(1-P_1)P_2(1-P_3)  \notag \\
  {}& \quad + \left(\sum\nolimits_{k=n+1}^{\infty}\sum\nolimits_{y=k-n+2}^{\infty}\pi_{(y,k-n+1)}\right) p(1-P_1)P_2P_3 \bigg]\delta_1^{n-2}  \notag \\
={}&\left[\pi_{(1,0)}p(1-P_1)P_2 + \left(\sum\nolimits_{n=2}^{\infty}\pi_{(n,1)}\right) p(1-P_1)P_2P_3 \right]\delta_1^{n-2}  \notag \\
  {}& \quad + \Big[ \left(\sum\nolimits_{n=2}^{\infty}\pi_{(n,0)}\right) p(1-P_1)P_2 + \left( \sum\nolimits_{n=2}^{\infty} \sum\nolimits_{m=1}^{n-1}\pi_{(n,m)}\right) p(1-P_1)P_2(1-P_3)  \notag \\
  {}& \quad + \left(\sum\nolimits_{m=2}^{\infty}\sum\nolimits_{n=m+1}^{\infty}\pi_{(n,m)}\right) p(1-P_1)P_2P_3 \Big]\delta_1^{n-2} \\
={}& \left[\left(\sum\nolimits_{n=1}^{\infty}\pi_{(n,0)}\right)p(1-P_1)P_2  + \left( \sum\nolimits_{m=1}^{\infty} \sum\nolimits_{n=m+1}^{\infty}\pi_{(n,m)}\right) p(1-P_1)P_2 \right]\delta_1^{n-2}   \\
={}& p(1-P_1)P_2\delta_1^{n-2}
\end{align}

Let $k-n+1=\tilde n$, $j=\tilde m$, then
\begin{equation}
\sum\nolimits_{k=n+1}^{\infty} \sum\nolimits_{j=1}^{k-n}\pi_{(k-n+1,j)}
=\sum\nolimits_{\tilde n=2}^{\infty}\sum\nolimits_{\tilde m=1}^{\tilde n-1}\pi_{(\tilde n,\tilde m)}  \notag
\end{equation}

Similarly, do the substitutions $y=\tilde n$ and $k-n+1=\tilde m$, it shows that
\begin{equation}
\sum\nolimits_{k=n+1}^{\infty}\sum\nolimits_{y=k-n+2}^{\infty}\pi_{(y,k-n+1)}
=\sum\nolimits_{\tilde m=2}^{\infty}\sum\nolimits_{\tilde n=\tilde m+1}^{\infty}\pi_{(\tilde n,\tilde m)}  \notag
\end{equation}

Substituting these sums, we obtain the equation (36). In addition, observing that
\begin{align}
&\sum\nolimits_{n=2}^{\infty}\pi_{(n,1)} +  \sum\nolimits_{m=2}^{\infty} \sum\nolimits_{n=m+1}^{\infty}\pi_{(n,m)} = \sum\nolimits_{m=1}^{\infty} \sum\nolimits_{n=m+1}^{\infty}\pi_{(n,m)}  \notag
\end{align}
and
\begin{equation}
\sum\nolimits_{m=1}^{\infty} \sum\nolimits_{n=m+1}^{\infty}\pi_{(n,m)}
=\sum\nolimits_{n=2}^{\infty} \sum\nolimits_{m=1}^{n-1}\pi_{(n,m)}   \notag
\end{equation}
which yields the result (37). Equation (38) holds because the probabilities contained in bracket of (37) add up to 1.

Finally, we derive the following recursive equation
\begin{equation}
\pi_{(n,0)}=\pi_{(n-1,0)}\delta_2  + p(1-P_1)P_2\delta_3\delta_1^{n-2}  \qquad  (n\geq2)
\end{equation}

Applying (39) iteratively, the general formula of $\pi_{(n,0)}$ can be obtained. We show that
\begin{align}
\pi_{(n,0)}&=\pi_{(1,0)}\delta_2^{n-1}+p(1-P_1)P_2\delta_3 \sum\nolimits_{j=0}^{n-2}\delta_2^j\delta_1^{n-2-j}  \notag \\
&=\pi_{(1,0)}\delta_2^{n-1}+p(1-P_1)P_2\delta_3 \frac{\delta_1^{n-1}-\delta_2^{n-1}}{\delta_1-\delta_2}
\end{align}

The first probability $\pi_{(1,0)}$ is equal to $pP_1$, which can be determined directly from the last equation of (33). Remember that $\delta_2=1-p(P_1+P_2-P_1P_2)$ and
\begin{equation}
\delta_1=(1-P_3)\delta_2, \quad  \delta_3=P_3\delta_2  \notag
\end{equation}
we have
\begin{align}
\pi_{(n,0)}&=pP_1\delta_2^{n-1}+p(1-P_1)P_2\delta_3\frac{\delta_1^{n-1}-\delta_2^{n-1}}{(1-P_3)\delta_2-\delta_2} \notag \\
&=pP_1\delta_2^{n-1}+p(1-P_1)P_2\left(\delta_2^{n-1}-\delta_1^{n-1} \right)  \notag \\
&=\left[pP_1+p(1-P_1)P_2\right]\delta_2^{n-1} - p(1-P_1)P_2\delta_1^{n-1}  \notag \\
&=(1-\delta_2)\delta_2^{n-1} - p(1-P_1)P_2[(1-P_3)\delta_2]^{n-1}
\end{align}

Let $n=1$, equation (41) gives $\pi_{(1,0)}=pP_1$. Therefore, we prove that expression (41) is valid for all $n\geq1$.

Next, probabilities $\pi_{(n,m)}$, $n>m\geq1$ are solved. We first calculate $\pi_{(n,1)}$, then by the relation
\begin{equation}
\pi_{(n,m)}=\pi_{(n-m+1,1)}\delta_1^{m-1}
\end{equation}
all the probabilities $\pi_{(n,m)}$ are obtained.

When $n\geq3$, from the second line of equations (33), we show that
\begin{align}
\pi_{(n,1)}&= \pi_{(n-1,0)}p(1-P_1)P_2 + \sum\nolimits_{j=1}^{n-2}\pi_{(n-1,j)}p(1-P_1)P_2(1-P_3) + \sum\nolimits_{k=n}^{\infty}\pi_{(k,n-1)}p(1-P_1)P_2P_3  \notag \\
&=\left\{(1-\delta_2)\delta_2^{n-2} - p(1-P_1)P_2\delta_1^{n-2}\right\}p(1-P_1)P_2  \notag \\
& \qquad  +  \sum\nolimits_{j=1}^{n-2}\pi_{(n-j,1)}\delta_1^{j-1} p(1-P_1)P_2(1-P_3) + p(1-P_1)P_2\delta_1^{n-2}p(1-P_1)P_2P_3 \\
&= p(1-P_1)P_2(1-\delta_2)\delta_2^{n-2} - [p(1-P_1)P_2]^2(1-P_3)\delta_1^{n-2}  \notag \\
& \qquad  + p(1-P_1)P_2(1-P_3)\sum\nolimits_{j=1}^{n-2}\pi_{(n-j,1)}\delta_1^{j-1}
\end{align}

In (43), we have substituted equations (41) and (38). Using relation (42), we can rewrite $\pi_{(n-1,j)}$ as $\pi_{(n-j,1)}\delta_1^{j-1}$.

Computing the difference
\begin{align}
\pi_{(n,1)}-\pi_{(n-1,1)}\delta_1 = p(1-P_1)P_2(1-\delta_2)\delta_2^{n-3}(\delta_2-\delta_1) + p(1-P_1)P_2(1-P_3)\pi_{(n-1,1)} \notag
\end{align}
which is equivalent to
\begin{align}
\pi_{(n,1)}&=\pi_{(n-1,1)}[\delta_1+p(1-P_1)P_2(1-P_3)]  + p(1-P_1)P_2(1-\delta_2)P_3\delta_2^{n-2} \notag \\
&= \pi_{(n-1,1)}(1-P_3)(1-pP_1) + p(1-P_1)P_2(1-\delta_2)P_3\delta_2^{n-2}
\end{align}
where
\begin{align}
\delta_1 + p(1-P_1)P_2(1-P_3) &= (1-P_3)\delta_2 + p(1-P_1)P_2(1-P_3)  \notag \\
&=(1-P_3)\left\{1-p(P_1+P_2-P_1P_2) + p(1-P_1)P_2 \right\}  \notag \\
&=(1-P_3)(1-pP_1)     \notag
\end{align}

Equation (45) gives a recursive relation of the probabilities $\pi_{(n,1)}$, $n\geq3$, from which we can derive that
\begin{align}
\pi_{(n,1)}&= \pi_{(2,1)}[(1-P_3)(1-pP_1)]^{n-2} + p(1-P_1)P_2(1-\delta_2)P_3 \sum\nolimits_{j=0}^{n-3}[(1-P_3)(1-pP_1)]^j \delta_2^{n-2-j}   \notag \\
&=\pi_{(2,1)}[(1-P_3)(1-pP_1)]^{n-2} + p(1-P_1)P_2P_3(1-\delta_2)\delta_2  \frac{\delta_2^{n-2}-[(1-P_3)(1-pP_1)]^{n-2}}{\delta_2-(1-P_3)(1-pP_1)}   \\
&=\frac{p(1-P_1)P_2P_3(1-\delta_2)\delta_2}{\delta_2-(1-P_3)(1-pP_1)}\delta_2^{n-2}+ \left(\pi_{(2,1)} - \frac{p(1-P_1)P_2P_3(1-\delta_2)\delta_2}{\delta_2-(1-P_3)(1-pP_1)} \right) [(1-P_3)(1-pP_1)]^{n-2}
\end{align}

Using the basic relation in (33), the probability $\pi_{(2,1)}$ is calculated as
\begin{align}
\pi_{(2,1)}&=\pi_{(1,0)}p(1-P_1)P_2 + \sum\nolimits_{k=2}^{\infty}\pi_{(k,1)}p(1-P_1)P_2P_3  \notag \\
&=pP_1p(1-P_1)P_2 + [p(1-P_1)P_2]^2P_3    \notag \\
&=p^2(1-P_1)P_2[P_1+(1-P_1)P_2P_3]  \notag
\end{align}

In equation (46), it is easy to see that for the case $n=2$, expression (47) reduces to $\pi_{(2,1)}=\pi_{(2,1)}$. Therefore, we show that (47) is actually valid for all $n\geq2$.

Now, for $n>m\geq1$, the probabilities $\pi_{(n,m)}$ can be obtained eventually. It shows that
\begin{align}
\pi_{(n,m)}&=\pi_{(n-m+1,1)}\delta_1^{m-1} \notag \\
&= \bigg\{ \frac{p(1-P_1)P_2P_3(1-\delta_2)\delta_2}{\delta_2-(1-P_3)(1-pP_1)}\delta_2^{n-m-1} \notag \\
& \quad + \left(\pi_{(2,1)} - \frac{p(1-P_1)P_2P_3(1-\delta_2)\delta_2}{\delta_2-(1-P_3)(1-pP_1)} \right) [(1-P_3)(1-pP_1)]^{n-m-1} \bigg\}\delta_1^{m-1}  \notag \\
&= \frac{p(1-P_1)P_2P_3(1-\delta_2)\delta_2}{\delta_2-(1-P_3)(1-pP_1)}\delta_2^{n-2}(1-P_3)^{m-1}   \notag \\
& \quad + \left(\pi_{(2,1)} - \frac{p(1-P_1)P_2P_3(1-\delta_2)\delta_2}{\delta_2-(1-P_3)(1-pP_1)} \right) [(1-P_3)(1-pP_1)]^{n-2}\left(\frac{\delta_2}{1-pP_1}\right)^{m-1}
\end{align}

So far, we have obtained all the stationary probabilities of process $AoI_{NB}^P$ and complete the proof of Theorem 2.

\section{Proof of Theorem 2}
In this part, according to equation (4), we calculation the stationary distribution of AoI.

For $n\geq2$, we first compute
\begin{align}
&\sum\nolimits_{m=1}^{n-1}\pi_{(n,m)} \notag \\
={}&\sum\nolimits_{m=1}^{n-1}\Bigg\{\pi_{(2,1)}[(1-pP_1)(1-P_3)]^{n-2}\left(\frac{\delta}{1-pP_1}\right)^{m-1} + \frac{p(1-P_1)P_2P_3(1-\delta)\delta}{\delta-(1-pP_1)(1-P_3)} \notag \\
  {}& \quad  \times \left[ \delta^{n-2}(1-P_3)^{m-1} - [(1-pP_1)(1-P_3)]^{n-2}\left(\frac{\delta}{1-pP_1} \right)^{m-1} \right] \Bigg\}  \notag \\
={}&\pi_{(2,1)}\frac{[(1-pP_1)(1-P_3)]^{n-1}-[(1-P_3)\delta]^{n-1}}{(1-pP_1-\delta)(1-P_3)} + \frac{p(1-P_1)P_2P_3(1-\delta)\delta}{\delta-(1-pP_1)(1-P_3)} \notag \\
  {}& \quad  \times \left\{ \frac{\delta^{n-2}-(1-P_3)[(1-P_3)\delta]^{n-2}}{P_3} - \frac{[(1-pP_1)(1-P_3)]^{n-1}-[(1-P_3)\delta]^{n-1}}{(1-pP_1-\delta)(1-P_3)}  \right\}  \notag \\
={}& \frac{p(1-P_1)P_2(1-\delta)}{\delta-(1-pP_1)(1-P_3)}\left\{\delta^{n-1}-[(1-P_3)\delta]^{n-1}\right\}  \notag \\
  {}& \quad  + \left(\pi_{(2,1)} - \frac{p(1-P_1)P_2P_3(1-\delta)\delta}{\delta-(1-pP_1)(1-P_3)} \right) \frac{[(1-pP_1)(1-P_3)]^{n-1}-[(1-P_3)\delta]^{n-1}}{(1-pP_1-\delta)(1-P_3)}   \\
={}&\eta_1\left\{\delta^{n-1}-[(1-P_3)\delta]^{n-1}\right\} + \eta_2\left\{[(1-pP_1)(1-P_3)]^{n-1}-[(1-P_3)\delta]^{n-1}\right\} \notag
\end{align}
where we denote
\begin{align}
\eta_1&=\frac{p(1-P_1)P_2(1-\delta)}{\delta-(1-pP_1)(1-P_3)},  \notag \\
\eta_2&=\left(\pi_{(2,1)}-\frac{p(1-P_1)P_2P_3(1-\delta)\delta}{\delta-(1-pP_1)(1-P_3)}\right) \frac{1}{(1-pP_1-\delta)(1-P_3)}  \notag
\end{align}

Therefore, when $n\geq2$, the probability that AoI equals $n$ is calculated as
\begin{align}
\Pr\{\Delta=n\}&=\pi_{(n,0)}+\sum\nolimits_{m=1}^{n-1}\pi_{(n,m)}  \notag  \\
&= [(1-\delta)+\eta_1]\delta^{n-1}+\eta_2[(1-pP_1)(1-P_3)]^{n-1}  \notag \\
&\qquad \quad -\left(\eta_1+\eta_2+p(1-P_1)P_2\right)[(1-P_3)\delta]^{n-1}
\end{align}

To obtain equation (50), we have substituted the probability expression (2).

Since
\begin{align}
\delta-(1-pP_1)(1-P_3) &=1-p(P_1+P_2-P_1P_2)-(1-pP_1)(1-P_3)  \notag \\
&=1-pP_1-p(1-P_1)P_2-(1-pP_1)(1-P_3)  \notag \\
&=(1-pP_1)P_3-p(1-P_1)P_2   \notag
\end{align}
we have the first coefficient
\begin{align}
(1-\delta)+\eta_1&=1-\delta + \frac{p(1-P_1)P_2(1-\delta)}{\delta-(1-pP_1)(1-P_3)} \notag \\
&=(1-\delta) \left( 1 + \frac{p(1-P_1)P_2}{(1-pP_1)P_3-p(1-P_1)P_2}\right) \notag \\
&=\frac{(1-pP_1)P_3(1-\delta)}{(1-pP_1)P_3-p(1-P_1)P_2}  \notag
\end{align}

Notice that
\begin{align}
1-pP_1-\delta =1-pP_1-[1-p(P_1+P_2-P_1P_2)] = p(1-P_1)P_2   \notag
\end{align}

Substituting $\pi_{(2,1)}$, it shows that
\begin{align}
\eta_2&=\frac{p[P_1+(1-P_1)P_2P_3]}{1-P_3} - \frac{P_3(1-\delta)\delta}{[(1-pP_1)P_3-p(1-P_1)P_2](1-P_3)}  \notag
\end{align}

In following paragraph, we will prove that the coefficient before the last term, i.e., $\eta_1+\eta_2+p(1-P_1)P_2$ is zero. Thus, for $n\geq2$ we obtain
\begin{align}
\Pr\{\Delta_{NB}^P=n\}&=\frac{(1-pP_1)P_3(1-\delta)}{(1-pP_1)P_3-p(1-P_1)P_2}\delta^{n-1} + \eta_2[(1-pP_1)(1-P_3)]^{n-1}
\end{align}

From equation (49), it is easy to see that the sum in (4) is zero for the case $n=1$. So, the expression (51) is valid for all $n\geq1$.

We now show that $\eta_1+\eta_2+p(1-P_1)P_2=0$. That is
\begin{align}
&\frac{p(1-P_1)P_2(1-\delta)}{\delta-(1-pP_1)(1-P_3)} + \left(p[P_1+(1-P_1)P_2P_3]-\frac{P_3(1-\delta)\delta}{\delta-(1-pP_1)(1-P_3)}\right) \notag \\
& \qquad \qquad \qquad \qquad \qquad \qquad \qquad \qquad  \times \frac{p(1-P_1)P_2}{(1-pP_1-\delta)(1-P_3)} + p(1-P_1)P_2 =0   \notag
\end{align}
which is equivalent to
\begin{align}
&\frac{1-\delta}{\delta-(1-pP_1)(1-P_3)}  \notag \\
&  + \left(p[P_1+(1-P_1)P_2P_3]-\frac{P_3(1-\delta)\delta}{\delta-(1-pP_1)(1-P_3)}\right) \frac{1}{(1-pP_1-\delta)(1-P_3)} + 1 =0   \notag \\
\Leftrightarrow{}& \frac{1-\delta}{\delta-(1-pP_1)(1-P_3)} + \frac{p[P_1+(1-P_1)P_2P_3]}{(1-pP_1-\delta)(1-P_3)} \notag \\
&  - \frac{P_3(1-\delta)\delta}{[\delta-(1-pP_1)(1-P_3)](1-pP_1-\delta)(1-P_3)} +1 =0
\end{align}

In (52), we have
\begin{equation}
\delta-(1-pP_1)(1-P_3)=(1-pP_1)P_3 - p(1-P_1)P_2   \notag
\end{equation}
and
\begin{align}
(1-pP_1-\delta)(1-P_3) =\left[1-pP_1- \left(1-p(P_1+P_2-P_1P_2) \right)\right](1-P_3) =p(1-P_1)P_2(1-P_3)  \notag
\end{align}

Then, equation (52) can be rewritten as
\begin{align}
&\frac{1-\delta}{(1-pP_1)P_3-p(1-P_1)P_2}\left(\frac{P_3\delta}{p(1-P_1)P_2(1-P_3)}-1\right)  =\frac{p[P_1+(1-P_1)P_2P_3]}{p(1-P_1)P_2(1-P_3)}+1
\end{align}
where the RHS of (53) is equal to
\begin{align}
&\frac{p[P_1+(1-P_1)P_2P_3]+p(1-P_1)P_2(1-P_3)}{p(1-P_1)P_2(1-P_3)} = \frac{p[P_1+(1-P_1)P_2]}{p(1-P_1)P_2(1-P_3)}= \frac{1-\delta}{p(1-P_1)P_2(1-P_3)}  \notag
\end{align}

Therefore, to prove (52), it suffices to show that
\begin{align}
&\frac{1}{(1-pP_1)P_3-p(1-P_1)P_2}\left(\frac{P_3\delta}{p(1-P_1)P_2(1-P_3)}-1\right) = \frac{1}{p(1-P_1)P_2(1-P_3)} \notag \\
\Leftrightarrow{}&\frac{P_3\delta}{p(1-P_1)P_2(1-P_3)}-1 = \frac{(1-pP_1)P_3-p(1-P_1)P_2}{p(1-P_1)P_2(1-P_3)} = \frac{(1-pP_1)P_3-p(1-P_1)P_2P_3}{p(1-P_1)P_2(1-P_3)}-1  \notag \\
\Leftrightarrow{}&\delta=(1-pP_1)-p(1-P_1)P_2 =1-p(P_1+P_2-P_1P_2)
\end{align}

Since equation (54) holds, we verify the original equation (52) and show that the coefficient $\eta_1+\eta_2+p(1-P_1)P_2$ is equal to zero. This completes the proof of Theorem 3.

\section{Proof of Theorem 5}
In this Appendix, we solve the system of stationary equations (6) and find all the stationary probabilities $\pi_{(n,m)}$ of AoI stochastic process $AoI_{NB}^{NP}$. For convenience, we copy these equations in the following.
\begin{equation}
\begin{cases}
\pi_{(n,m)}=\pi_{(n-1,m-1)}[(1-p)(1-P_3)+p(1-P_1)(1-P_3)]   & \qquad    (n>m\geq2)    \\
\pi_{(n,1)}=\pi_{(n-1,0)}p(1-P_1)P_2                  & \qquad   (n\geq2)   \\
\pi_{(n,0)}=\pi_{(n-1,0)}[(1-p)+p(1-P_1)(1-P_2)]    \\
\qquad \qquad      + \left(\sum\nolimits_{k=n}^{\infty}\pi_{(k,n-1)}\right)[(1-p)P_3+p(1-P_1)P_3]   & \qquad  (n\geq2)    \\
\pi_{(1,0)}= \left(\sum\nolimits_{n=1}^{\infty}\pi_{(n,0)} + \sum\nolimits_{n=2}^{\infty}\sum\nolimits_{m=1}^{n-1}\pi_{(n,m)}\right)pP_1
\end{cases}
\end{equation}

First of all, applying the first line of (55) repeatedly, we can write $\pi_{(n,m)}$ as
\begin{align}
\pi_{(n,m)}&=\pi_{(n-1,m-1)}[(1-pP_1)(1-P_3)]  \notag \\
&=\pi_{(n-2,m-2)}[(1-pP_1)(1-P_3)]^2  \notag \\
&\qquad \qquad \quad \quad \vdots  \notag \\
&=\pi_{(n-m+1,1)}[(1-pP_1)(1-P_3)]^{m-1}  \notag \\
&=\pi_{(n-m,0)}p(1-P_1)P_2[(1-pP_1)(1-P_3)]^{m-1}
\end{align}
where to obtain the last step (56), the second equation in (55) is applied.

Substituting (56), we first handle the sum in the third line of (55). It shows that
\begin{align}
\sum\nolimits_{k=n}^{\infty}\pi_{(k,n-1)}&= \sum\nolimits_{k=n}^{\infty}\pi_{(k-n+1,0)}p(1-P_1)P_2 [(1-pP_1)(1-P_3)]^{n-2}\notag \\
&= \left(\sum\nolimits_{k=1}^{\infty}\pi_{(k,0)}\right) p(1-P_1)P_2 [(1-pP_1)(1-P_3)]^{n-2}  \notag
\end{align}

Since all the probabilities $\pi_{(n,m)}$ have to add up to 1, we have the relation
\begin{align}
1&=\sum\nolimits_{n=1}^{\infty}\pi_{(n,0)} + \sum\nolimits_{n=m+1}^{\infty}\sum\nolimits_{m=1}^{\infty}\pi_{(n,m)} \notag \\
&=\sum\nolimits_{n=1}^{\infty}\pi_{(n,0)} + \sum\nolimits_{n=m+1}^{\infty}\sum\nolimits_{m=1}^{\infty}\pi_{(n-m,0)} p(1-P_1)P_2[(1-pP_1)(1-P_3)]^{m-1} \notag \\
&=\sum\nolimits_{n=1}^{\infty}\pi_{(n,0)} + \left(\sum\nolimits_{t=1}^{\infty}\pi_{(t,0)}\right) p(1-P_1)P_2 \sum\nolimits_{m=1}^{\infty} [(1-pP_1)(1-P_3)]^{m-1}  \\
&=\left(\sum\nolimits_{k=1}^{\infty}\pi_{(k,0)}\right)\left[1 + \frac{p(1-P_1)P_2}{1-(1-pP_1)(1-P_3)}\right]
\end{align}

To obtain (57), do the substitution $t=n-m$. From equation (58) we can derive that
\begin{equation}
\sum\nolimits_{k=1}^{\infty}\pi_{(k,0)}=\frac{1-(1-pP_1)(1-P_3)}{1-(1-pP_1)(1-P_3)+p(1-P_1)P_2} \notag
\end{equation}

Combining above results, for $n\geq2$ we can write that
\begin{equation}
\pi_{(n,0)}=\pi_{(n-1,0)}\delta + \xi \left[(1-pP_1)(1-P_3)\right]^{n-1}
\end{equation}
where
\begin{align}
\delta&=1-p(P_1+P_2-P_1P_2),   \notag \\
\xi&=\frac{p(1-P_1)P_2P_3[1-(1-pP_1)(1-P_3)]}{[1-(1-pP_1)(1-P_3)+p(1-P_1)P_2](1-P_3)} \notag
\end{align}

Equation (59) gives a recursive relation of the probabilities $\pi_{(n,0)}$. Using (59) iteratively yields that
\begin{align}
\pi_{(n,0)}&=\pi_{(1,0)}\delta^{n-1} + \xi \sum\nolimits_{j=0}^{n-2}\delta^j[(1-pP_1)(1-P_3)]^{n-1-j}   \notag \\
&=\pi_{(1,0)}\delta^{n-1} + \frac{\xi(1-pP_1)(1-P_3)}{(1-pP_1)(1-P_3)-\delta} \left\{[(1-pP_1)(1-P_3)]^{n-1} - \delta^{n-1}\right\}  \\
&=\beta_1 \delta^{n-1} + \beta_2[(1-pP_1)(1-P_3)]^{n-1}
\end{align}

Substituting $\xi$, it show that the coefficient $\beta_2$ is equal to
\begin{align}
\beta_2&=\frac{[1-(1-pP_1)(1-P_3)]p(1-P_1)P_2P_3}{1-(1-pP_1)(1-P_3)+p(1-P_1)P_2} \frac{1-pP_1}{(1-pP_1)(1-P_3)-\delta}    \notag
\end{align}

Since $\pi_{(1,0)}=pP_1$, we have $\beta_1=pP_1-\beta_2$. It is easy to see that when $n=1$, equation (61) reduces to $\pi_{(1,0)}=\pi_{(1,0)}$, thus probability expression (61) is actually valid for all $n\geq1$.

According to (56), for $n>m\geq1$, the probability $\pi_{(n,m)}$ is then determined as
\begin{align}
\pi_{(n,m)}&= \pi_{(n-m,0)}p(1-P_1)P_2[(1-pP_1)(1-P_3)]^{m-1} \notag \\
={}& p(1-P_1)P_2\left\{ \beta_1\delta^{n-m-1} + \beta_2[(1-pP_1)(1-P_3)]^{n-m-1} \right\} [(1-pP_1)(1-P_3)]^{m-1} \notag \\
={}& p(1-P_1)P_2 \left\{\beta_1\delta^{n-2}\left(\frac{(1-pP_1)(1-P_3)}{\delta}\right)^{m-1} + \beta_2[(1-pP_1)(1-P_3)]^{n-2} \right\}
\end{align}

So far, we have completely solved the system of equations (55) and obtained all the stationary probabilities $\pi_{(n,m)}$.

\section{Proof of Theorem 8}
In this appendix, we determine all the stationary probabilities $\pi_{(n,m,l)}$ of stochastic process $AoI_B^P$ by solving the following system of equations.
\begin{equation}
\begin{cases}
\pi_{(n,m,l)}=\pi_{(n-1,m-1,l-1)}\delta(1-P_3)     & \qquad   (n>m>l\geq2)    \\
\pi_{(n,m,1)}=\left(\sum\nolimits_{j=0}^{m-2}\pi_{(n-1,m-1,j)}\right)\eta(1-P_3)     & \qquad  (n>m\geq2)    \\
\pi_{(n,m,0)}=\pi_{(n-1,m-1,0)}\delta(1-P_3) +\left(\sum\nolimits_{k=n}^{\infty}\pi_{(k,n-1,m-1)}\right)\delta P_3  & \qquad   (n>m\geq2)  \\
\pi_{(n,1,0)}=\pi_{(n-1,0,0)}\eta +\left(\sum\nolimits_{k=n}^{\infty}\sum\nolimits_{j=0}^{n-2}\pi_{(k,n-1,j)}\right)\eta P_3 & \qquad  (n\geq2)    \\
\pi_{(n,0,0)}=\pi_{(n-1,0,0)}\delta+\left(\sum\nolimits_{k=n}^{\infty}\pi_{(k,n-1,0)}\right)\delta P_3   &  \qquad  (n\geq2)  \\
\pi_{(1,0,0)}=\Big(\sum\nolimits_{n=1}^{\infty}\pi_{(n,0,0)}+ \sum\nolimits_{n=2}^{\infty}\sum\nolimits_{m=1}^{n-1}\pi_{(n,m,0)}\\
\qquad  \qquad  \qquad  \qquad \qquad  + \sum\nolimits_{n=3}^{\infty}\sum\nolimits_{m=2}^{n-1}\sum\nolimits_{l=1}^{m-1}\pi_{(n,m,l)}\Big)pP_1
\end{cases}
\end{equation}

First of all, for $n>m>l\geq1$, using first two lines of (63) we can obtain that
\begin{align}
\pi_{(n,m,l)}&= \pi_{(n-1,m-1,l-1)}[\delta(1-P_3)]  \notag \\
  {}& \qquad \qquad \qquad  \vdots    \notag \\
&= \pi_{(n-l+1,m-l+1,1)}[\delta(1-P_3)]^{l-1}  \notag \\
&= \left(\sum\nolimits_{j=0}^{m-l-1}\pi_{(n-l,m-l,j)}\right)\eta(1-P_3)[\delta(1-P_3)]^{l-1}
\end{align}

Substituting expression (64), the sum in the third equation of (63) can be rewritten as
\begin{align}
\sum\nolimits_{k=n}^{\infty}\pi_{(k,n-1,m-1)}&= \sum\nolimits_{k=n}^{\infty}\left(\sum\nolimits_{j=0}^{n-m-1}\pi_{(k-m+1,n-m,j)} \right) \eta(1-P_3)[\delta(1-P_3)]^{m-2} \notag \\
&= \left( \sum\nolimits_{k=n-m+1}^{\infty}\sum\nolimits_{j=0}^{n-m-1}\pi_{(k,n-m,j)}\right) \eta(1-P_3)[\delta(1-P_3)]^{m-2}\notag \\
&= t_{n-m}\eta(1-P_3)[\delta(1-P_3)]^{m-2}
\end{align}
where we define
\begin{equation}
t_n= \sum\nolimits_{k=n+1}^{\infty}\sum\nolimits_{j=0}^{n-1}\pi_{(k,n,j)}  \qquad  (n\geq1)
\end{equation}
which is actually the probability that the middle parameter equals $n$.

Therefore, for $n>m\geq2$ we have
\begin{align}
\pi_{(n,m,0)}&=\pi_{(n-1,m-1,0)}\delta(1-P_3)  + t_{n-m}\eta (1-P_3)[\delta(1-P_3)]^{m-2} \delta P_3  \notag  \\
&=\pi_{(n-1,m-1,0)}\delta(1-P_3) + t_{n-m}\eta P_3[\delta(1-P_3)]^{m-1}
\end{align}

Equation (67) is a recursive formula of the probabilities $\pi_{(n,m,0)}$. By repeatedly using (67), it shows that
\begin{align}
\pi_{(n,m,0)}&=\pi_{(n-m+1,1,0)}[\delta(1-P_3)]^{m-1} + (m-1)t_{n-m}\eta P_3[\delta(1-P_3)]^{m-1}  \notag  \\
&= \left\{\pi_{(n-m,0,0)}\eta+t_{n-m}\eta P_3\right\}[\delta(1-P_3)]^{m-1} + (m-1)t_{n-m}\eta P_3[\delta(1-P_3)]^{m-1}   \\
&=\pi_{(n-m,0,0)}\eta[\delta(1-P_3)]^{m-1} + m t_{n-m}\eta P_3[\delta(1-P_3)]^{m-1}
\end{align}

In (68), we have used the fourth line in (63) to represent $\pi_{(n-m+1,1,0)}$ with $\pi_{(n-m,0,0)}$ and $t_{n-m}$. Notice that equation (69) is valid for all $n>m\geq1$.

Substituting equation (69), the summation within the fifth line of (63) can be transformed as follows.
\begin{align}
\sum\nolimits_{k=n}^{\infty}\pi_{(k,n-1,0)}&=\sum\nolimits_{k=n}^{\infty}\Big\{\pi_{(k-n+1,0,0)}\eta[\delta(1-P_3)]^{n-2} +(n-1)t_{k-n+1}\eta P_3[\delta(1-P_3)]^{n-2} \Big\}  \notag \\
&= \left(\sum\nolimits_{k=1}^{\infty}\pi_{(k,0,0)}\right)\eta[\delta(1-P_3)]^{n-2} + \left(\sum\nolimits_{k=1}^{\infty}t_{k}\right) \eta P_3(n-1)[\delta(1-P_3)]^{n-2}   \notag
\end{align}

Since $t_n$ denotes the probability that the middle parameter equals $n$, we have the relation
\begin{equation}
t_0 + \sum\nolimits_{n=1}^{\infty}t_n =1   \notag
\end{equation}
in which
\begin{equation}
t_0=\sum\nolimits_{k=1}^{\infty}\pi_{(k,0,0)}=S   \notag
\end{equation}
is the probability that the middle parameter takes value 0, which is represented by $S$ in following paragraphs.

Combining above results, for $n\geq2$ we have that
\begin{align}
\pi_{(n,0,0)}&=\pi_{(n-1,0,0)}\delta + \Big\{S\eta[\delta(1-P_3)]^{n-2} + (1-S)\eta P_3(n-1)[\delta(1-P_3)]^{n-2}\Big\}\delta P_3
\end{align}
which gives a recursive formula for $\pi_{(n,0,0)}$. Before the general expression of $\pi_{(n,0,0)}$ is given, we first determine $S$.

From $n=2$ to $\infty$, adding up both sides of equation (70) yields that
\begin{equation}
S-\pi_{(1,0,0)}=S\delta+\frac{S\delta\eta P_3}{1-\delta(1-P_3)}+\frac{(1-S)\delta\eta P_3^2}{[1-\delta(1-P_3)]^2}
\end{equation}

Since we have known that $\pi_{(1,0,0)}=pP_1$, from (71) we can derive
\begin{equation}
S=\frac{pP_1[1-\delta(1-P_3)]^2+\delta\eta P_3^2}{(1-\delta)[1-\delta(1-P_3)]^2-\delta\eta P_3(1-\delta)(1-P_3)}
\end{equation}

Applying equation (70) iteratively, we have
\begin{align}
\pi_{(n,0,0)}&=\pi_{(1,0,0)}\delta^{n-1}+S\delta\eta P_3\sum\nolimits_{j=0}^{n-2}\delta^j[\delta(1-P_3)]^{n-2-j} \notag \\
&\qquad +(1-S)\delta\eta P_3^2\sum\nolimits_{j=0}^{n-2}\delta^j(n-1-j)[\delta(1-P_3)]^{n-2-j} \notag \\
&= \left(pP_1+\eta\right)\delta^{n-1}-\eta[\delta(1-P_3)]^{n-1}  - (1-S)\eta P_3(n-1)[\delta(1-P_3)]^{n-1}
\end{align}
in which we have substituted $\pi_{(1,0,0)}=pP_1$ and omit some intermediate calculation details. Equation (73) is valid for $n=1$ can be verified directly.

In order to determine probabilities $\pi_{(n,m,0)}$ by (69), we next derive the general expression of $t_n$. For $n\geq2$, it shows that
\begin{align}
t_n&=\sum\nolimits_{k=n+1}^{\infty}\sum\nolimits_{j=0}^{n-1}\pi_{(k,n,j)}  \notag \\
&=\sum\nolimits_{k=n+1}^{\infty}\pi_{(k,n,0)}+\sum\nolimits_{k=n+1}^{\infty}\sum\nolimits_{j=1}^{n-1}\pi_{(k,n,j)}  \notag \\
&=\sum\nolimits_{k=n+1}^{\infty}\Big\{\pi_{(k-n,0,0)}\eta[\delta(1-P_3)]^{n-1}+n t_{k-n}\eta P_3[\delta(1-P_3)]^{n-1}\Big\} \notag \\
& \qquad +\sum\nolimits_{k=n+1}^{\infty}\sum\nolimits_{j=1}^{n-1}\left(\sum\nolimits_{y=0}^{n-j-1}\pi_{(k-j,n-j,y)}\right)\eta(1-P_3) [\delta(1-P_3)]^{j-1}   \\
&=S\eta [\delta(1-P_3)]^{n-1}+(1-S)\eta P_3n[\delta(1-P_3)]^{n-1}   \notag \\
&\qquad + \eta(1-P_3)\sum\nolimits_{j=1}^{n-1}t_{n-j}[\delta(1-P_3)]^{j-1}
\end{align}
where in (74) we have used equations (69) and (64). To obtain (75), notice that
\begin{align}
\sum\nolimits_{k=n+1}^{\infty}\sum\nolimits_{y=0}^{n-j-1}\pi_{(k-j,n-j,y)}=\sum\nolimits_{k=n-j+1}^{\infty}\sum\nolimits_{y=0}^{n-j-1}\pi_{(k,n-j,y)}= t_{n-j} \notag
\end{align}

Similarly, we also have
\begin{align}
t_{n-1}&=S\eta [\delta(1-P_3)]^{n-2}+(1-S)\eta P_3(n-1)[\delta(1-P_3)]^{n-2} \notag \\
& \quad + \eta(1-P_3)\sum\nolimits_{j=1}^{n-2}t_{n-1-j}[\delta(1-P_3)]^{j-1}
\end{align}

Compute the following difference
\begin{equation}
t_n-t_{n-1}\delta(1-P_3)=(1-S)\eta P_3[\delta(1-P_3)]^{n-1}+t_{n-1}\eta(1-P_3)  \notag
\end{equation}
which is equivalent to
\begin{equation}
t_n=t_{n-1}(\delta+\eta)(1-P_3)+(1-S)\eta P_3[\delta(1-P_3)]^{n-1}
\end{equation}

Applying equation (77) repeatedly yields that
\begin{align}
t_n&=t_1 [(\delta+\eta)(1-P_3)]^{n-1} +(1-S)\eta P_3 \sum\nolimits_{j=0}^{n-2}[(\delta+\eta)(1-P_3)]^j[\delta(1-P_3)]^{n-1-j}  \notag \\
&=t_1 [(\delta+\eta)(1-P_3)]^{n-1} +(1-S)\delta P_3 \left\{[(\delta+\eta)(1-P_3)]^{n-1} -  [\delta(1-P_3)]^{n-1}  \right\}
\end{align}

Using the general expression (69), we have
\begin{align}
t_1=\sum\nolimits_{k=2}^{\infty}\pi_{(k,1,0)} =\sum\nolimits_{k=2}^{\infty} \left\{\pi_{(k-1,0,0)}\eta + t_{k-1}\eta P_3 \right\}  =S\eta+(1-S)\eta P_3
\end{align}

Combining (78) and (79), eventually we derive that
\begin{align}
t_n&=\left\{S\eta+(1-S)(\delta+\eta)P_3\right\}[(\delta+\eta)(1-P_3)]^{n-1} - (1-S)\delta P_3[\delta(1-P_3)]^{n-1}  \notag \\
&=\widetilde{S}[(\delta+\eta)(1-P_3)]^{n-1} - (1-S)\delta P_3[\delta(1-P_3)]^{n-1}
\end{align}
where we denote
\begin{equation}
\widetilde{S}=S\eta+(1-S)(\delta+\eta)P_3   \notag
\end{equation}

Provided equations (73) and (80), now the probabilities $\pi_{(n,m,0)}$ can be determined by (69).
\begin{align}
\pi_{(n,m,0)}&= \pi_{(n-m,0,0)}\eta[\delta(1-P_3)]^{m-1}+ m t_{n-m}\eta P_3[\delta(1-P_3)]^{m-1}  \notag \\
&=\bigg\{ \left(pP_1+\eta\right)\delta^{n-m-1}-\eta[\delta(1-P_3)]^{n-m-1}  \notag \\
& \quad  -(1-S)\eta P_3 (n-m-1)[\delta(1-P_3)]^{n-m-1}\bigg\}\eta[\delta(1-P_3)]^{m-1}   \notag \\
& \quad +\left\{ \widetilde{S}[(\delta+\eta)(1-P_3)]^{n-m-1} - (1-S)\delta P_3 [\delta(1-P_3)]^{n-m-1} \right\}\eta P_3 m[\delta(1-P_3)]^{m-1}   \notag \\
&= \eta(pP_1+\eta)\delta^{n-2}(1-P_3)^{m-1}-\eta^2[\delta(1-P_3)]^{n-2} - (1-S)\eta^2 P_3(n-m-1)[\delta(1-P_3)]^{n-2}   \notag \\
& \quad + \widetilde{S}\eta P_3[(\delta+\eta)(1-P_3)]^{n-2}m\left(\frac{\delta}{\delta+\eta}\right)^{m-1} - (1-S)\delta\eta P_3^2 m[\delta(1-P_3)]^{n-2}
\end{align}

At last, the probabilities $\pi_{(n,m,l)}$, $n>m>l\geq1$ remains to be determined. Since
\begin{equation}
\pi_{(n,m,l)}=\pi_{(n-l+1,m-l+1,1)}[\delta(1-P_3)]^{l-1}  \notag
\end{equation}
thus it suffices to find all the probabilities $\pi_{(n,m,1)}$ for all the tuples $(n,m)$.

From the general formula (64), we have
\begin{align}
\pi_{(n,m,1)}&=\left(\sum\nolimits_{j=0}^{m-2}\pi_{(n-1,m-1,j)}\right)\eta(1-P_3)   \notag \\
&=\left(\pi_{(n-1,m-1,0)}+\sum\nolimits_{j=1}^{m-2}\pi_{(n-1,m-1,j)}\right)\eta(1-P_3)   \notag \\
&=\pi_{(n-1,m-1,0)}\eta(1-P_3) + \left(\sum\nolimits_{j=1}^{m-2}\pi_{(n-j,m-j,1)}[\delta(1-P_3)]^{j-1}\right)\eta(1-P_3)   \notag
\end{align}

Do once iteration, we can obtain the expansion expression of probability $\pi_{(n-1,m-1,1)}$. The difference
\begin{align}
&\pi_{(n,m,1)}-\pi_{(n-1,m-1,1)}\delta(1-P_3)  \notag \\
={}& \left\{\pi_{(n-1,m-1,0)}-\pi_{(n-2,m-2,0)}\delta(1-P_3)\right\}\eta(1-P_3) +\pi_{(n-1,m-1,1)}\eta(1-P_3)   \notag \\
={}&t_{n-m}\eta^2 P_3(1-P_3)[\delta(1-P_3)]^{m-2}+\pi_{(n,m,1)}\eta(1-P_3)
\end{align}
gives the following recursive relation
\begin{align}
\pi_{(n,m,1)}=\pi_{(n-1,m-1,1)}(\delta+\eta)(1-P_3) + t_{n-m}\eta^2 P_3(1-P_3)[\delta(1-P_3)]^{m-2}
\end{align}

In (82), we have substituted the probability expression (69). Then, the general formula of $\pi_{(n,m,1)}$ can be derived by applying (83) repeatedly. It shows that
\begin{align}
\pi_{(n,m,1)}&=\pi_{(n-m+2,2,1)}[(\delta+\eta)(1-P_3)]^{m-2}  \notag \\
&\quad + \eta^2P_3(1-P_3)t_{n-m}\sum\nolimits_{j=0}^{m-3}[(\delta+\eta)(1-P_3)]^j[\delta(1-P_3)]^{m-2-j}  \notag \\
&=\pi_{(n-m+2,2,1)}[(\delta+\eta)(1-P_3)]^{m-2}  \notag \\
&\quad + \delta\eta P_3(1-P_3)t_{n-m} \left\{[(\delta+\eta)(1-P_3)]^{m-2}-[\delta(1-P_3)]^{m-2}\right\}
\end{align}
in which
\begin{align}
\pi_{(n-m+2,2,1)}&= \pi_{(n-m+1,1,0)}\eta(1-P_3)  \notag \\
&=\left\{\pi_{(n-m,0,0)}\eta+t_{n-m}\eta P_3\right\}\eta(1-P_3) \notag \\
&=\pi_{(n-m,0,0)}\eta^2(1-P_3) + t_{n-m}\eta^2 P_3(1-P_3)
\end{align}

Combining equations (85) and (84) and merging the same terms gives
\begin{align}
\pi_{(n,m,1)}&=\pi_{(n-m,0,0)}\eta^2(1-P_3)[(\delta+\eta)(1-P_3)]^{m-2}  \notag \\
&\qquad + \eta P_3 t_{n-m} \left\{[(\delta+\eta)(1-P_3)]^{m-1}-[\delta(1-P_3)]^{m-1}\right\}  \notag
\end{align}

Since both $\pi_{(n-m,0,0)}$ and $t_{n-m}$ have been obtained, by substituting (73) and (80), we show that the final result is
\begin{align}
\pi_{(n,m,1)}&= \eta^2(1-P_3)\bigg\{(pP_1+\eta)\delta^{n-3}\left(\frac{(\delta+\eta)(1-P_3)}{\delta}\right)^{m-2} - \eta[\delta(1-P_3)]^{n-3}\left(\frac{\delta+\eta}{\delta}\right)^{m-2}  \notag \\
& \quad  -(1-S)\eta P_3(n-m-1)[\delta(1-P_3)]^{n-3}\left(\frac{\delta+\eta}{\delta}\right)^{m-2} \bigg\}  \notag \\
& \quad +\eta P_3\widetilde{S}[(\delta+\eta)(1-P_3)]^{n-2} - \eta P_3\widetilde{S}[(\delta+\eta)(1-P_3)]^{n-2}\left(\frac{\delta}{\delta+\eta}\right)^{m-1}  \notag \\
& \quad - \delta\eta P_3^2(1-S)[\delta(1-P_3)]^{n-2}\left(\frac{\delta+\eta}{\delta}\right)^{m-1} + \delta\eta P_3^2(1-S)[\delta(1-P_3)]^{n-2}
\end{align}

For $n>m>l\geq1$, the last step is
\begin{align}
\pi_{(n,m,l)}&=\pi_{(n-l+1,m-l+1,1)}[\delta(1-P_3)]^{l-1}  \notag \\
&=\eta^2(1-P_3)\bigg\{ \left(pP_1+\eta\right)\delta^{n-3}\left(\frac{(\delta+\eta)(1-P_3)}{\delta}\right)^{m-2} \left(\frac{\delta}{\delta+\eta}\right)^{l-1}  \notag \\
& \qquad  -\eta[\delta(1-P_3)]^{n-3}\left(\frac{\delta+\eta}{\delta}\right)^{m-2} \left(\frac{\delta}{\delta+\eta}\right)^{l-1} -(1-S)\eta P_3(n-m-1)   \notag \\
& \qquad  \times [\delta(1-P_3)]^{n-3}\left(\frac{\delta+\eta}{\delta}\right)^{m-2}\left(\frac{\delta}{\delta+\eta}\right)^{l-1} \bigg\}  \notag \\
& \quad + \eta P_3\widetilde{S}[(\delta+\eta)(1-P_3)]^{n-2}\left(\frac{\delta}{\delta+\eta}\right)^{l-1} - \eta P_3\widetilde{S}[(\delta+\eta)(1-P_3)]^{n-2}\left(\frac{\delta}{\delta+\eta}\right)^{m-1}  \notag \\
&\quad  -\delta\eta P_3^2(1-S)[\delta(1-P_3)]^{n-2}\left(\frac{\delta+\eta}{\delta}\right)^{m-1} \left(\frac{\delta}{\delta+\eta}\right)^{l-1}  \notag \\
&\quad + \delta\eta P_3^2(1-S)[\delta(1-P_3)]^{n-2}
\end{align}

Collecting the results in equations (73), (81) and (87), we show that all the stationary probabilities of stochastic process $AoI_B^P$ are determined. This completes the proof of Theorem 8.

\section{Proof of Theorem 9}
The stationary distribution of AoI, $\Delta_B^P$, is calculated in this Appendix. According to formula (15), for $n\geq3$, we have
\begin{align}
\Pr\{\Delta_B^P=n\}= \pi_{(n,0,0)}+\sum\nolimits_{m=1}^{n-1}\pi_{(n,m,0)}+\sum\nolimits_{m=2}^{n-1}\sum\nolimits_{l=1}^{m-1}\pi_{(n,m,l)} \notag
\end{align}
where $\pi_{(n,0,0)}$ is given in (16).

Substituting expression (17), the first sum is computed as
\begin{align}
&\sum\nolimits_{m=1}^{n-1}\pi_{(n,m,0)}  \notag \\
={}&\sum\nolimits_{m=1}^{n-1}\Big\{ \eta\left(pP_1+\eta\right)\delta^{n-2}\left(1-P_3\right)^{m-1}-\eta^2\left[\delta(1-P_3)\right]^{n-2} - (1-S)\eta^2P_3(n-m-1)  \notag \\
  {}& \times \left[\delta(1-P_3)\right]^{n-2} +\widetilde{S}\eta P_3 \left[(\delta+\eta)(1-P_3)\right]^{n-2} m \left(\frac{\delta}{\delta+\eta}\right)^{m-1} - (1-S)\delta\eta P_3^2 m[\delta(1-P_3)]^{n-2}  \Big\}   \notag \\
={}&\eta\left(pP_1+\eta\right)\delta^{n-2} \frac{1-(1-P_3)^{n-1}}{P_3}  -  \eta^2(n-1)\left[\delta(1-P_3)\right]^{n-2} - (1-S)\eta^2P_3  \notag \\
  {}& \times \frac{(n-1)(n-2)}{2}\left[\delta(1-P_3)\right]^{n-2} +\widetilde{S}\eta P_3 \left[(\delta+\eta)(1-P_3)\right]^{n-2}\frac{(\delta+\eta)^2}{\eta^2}      \notag \\
  {}&\times  \left[1- \left(\frac{\delta}{\delta+\eta}\right)^{n-1}- \frac{\eta}{\delta+\eta}(n-1) \left(\frac{\delta}{\delta+\eta}\right)^{n-1}  \right] -(1-S)\delta\eta P_3^2 \frac{n(n-1)}{2}[\delta(1-P_3)]^{n-2}
\end{align}

We omit the further calculations and directly show that
\begin{align}
\sum\nolimits_{m=1}^{n-1}\pi_{(n,m,0)}&=\frac{\eta(pP_1+\eta)}{\delta P_3}\delta^{n-1} + \frac{\widetilde{S}(\delta+\eta)P_3}{\eta(1-P_3)}[(\delta+\eta)(1-P_3)]^{n-1} \notag \\
&\qquad - \left\{\frac{\eta(pP_1+\eta)}{\delta P_3}+\frac{\widetilde{S}\delta P_3}{\eta(1-P_3)}-\frac{\eta^2-(1-S)\eta^2 P_3}{\delta(1-P_3)}\right\} [\delta(1-P_3)]^{n-1}  \notag \\
&\qquad - \left\{\frac{\eta^2-(1-S)\eta^2P_3}{\delta(1-P_3)}+\frac{\widetilde{S}\eta P_3}{\eta(1-P_3)} \right\}n[\delta(1-P_3)]^{n-1} \notag \\
&\qquad  -\frac{(1-S)\eta P_3(\delta P_3+\eta)}{2\delta(1-P_3)}n(n-1)[\delta(1-P_3)]^{n-1}
\end{align}

At last,
\begin{align}
&\sum\nolimits_{m=2}^{n-1}\sum\nolimits_{l=1}^{m-1}\pi_{(n,m,l)}  \notag \\
={}&\sum\nolimits_{m=2}^{n-1}\sum\nolimits_{l=1}^{m-1} \text{Equation (18)} \notag \\
={}&\sum\nolimits_{m=2}^{n-1}\Bigg\{\eta(pP_1+\eta)\delta^{n-2} \left[\left(\frac{(\delta+\eta)(1-P_3)}{\delta}\right)^{m-1}-(1-P_3)^{m-1}\right]   \notag \\
  {}&\qquad - \eta^2[\delta(1-P_3)]^{n-2}\left[\left(\frac{\delta+\eta}{\delta}\right)^{m-1} -1 \right]  \notag \\
  {}&\qquad \quad -(1-S)\eta^2 P_3 [\delta(1-P_3)]^{n-2} (n-m-1)\left[\left(\frac{\delta+\eta}{\delta}\right)^{m-1} -1\right]  \Bigg\}  \notag \\
  {}&\quad +\widetilde{S}(\delta+\eta)P_3[(\delta+\eta)(1-P_3)]^{n-2}\left[1-\left(\frac{\delta}{\delta+\eta}\right)^{m-1} \right] \notag \\
  {}&\quad -\widetilde{S}\eta P_3[(\delta+\eta)(1-P_3)]^{n-2}(m-1)\left(\frac{\delta}{\delta+\eta}\right)^{m-1} \notag \\
  {}&\quad -(1-S)\delta(\delta+\eta)P_3^2[\delta(1-P_3)]^{n-2}\left[\left(\frac{\delta+\eta}{\delta}\right)^{m-1}-1 \right]  \notag \\
  {}&\quad +(1-S)\delta\eta P_3^2(m-1)[\delta(1-P_3)]^{n-2}  \notag \\
={}&\eta(pP_1+\eta)\delta^{n-2} \bigg[ \frac{\eta(1-P_3)}{[\delta-(\delta+\eta)(1-P_3)]P_3}  \notag \\
  {}& \quad - \frac{\delta}{\delta-(\delta+\eta)(1-P_3)} \left(\frac{(\delta+\eta)(1-P_3)}{\delta}\right)^{n-1} +\frac{(1-P_3)^{n-1}}{P_3}\bigg]  \notag \\
  {}&-\eta^2[\delta(1-P_3)]^{n-2}\left[ \frac{\delta}{\eta}\left(\frac{\delta+\eta}{\delta}\right)^{n-1}-\frac{\delta+\eta}{\eta}-(n-2) \right]  \notag \\
  {}&-(1-S)\eta^2 P_3 [\delta(1-P_3)]^{n-2}\bigg[-(n-2)\frac{\delta+\eta}{\eta}  \notag \\
  {}& \quad + \frac{\delta^2}{\eta^2}\left(\frac{\delta+\eta}{\delta}\right)^{n-1}-\frac{\delta(\delta+\eta)}{\eta^2} -\frac{(n-2)(n-3)}{2}\bigg]   \notag \\
  {}&+\widetilde{S}(\delta+\eta)P_3[(\delta+\eta)(1-P_3)]^{n-2} \left[(n-2)-\frac{\delta}{\eta}+\frac{\delta+\eta}{\eta} \left( \frac{\delta}{\delta+\eta}\right)^{n-1} \right]  \notag \\
  {}&-\widetilde{S}\eta P_3[(\delta+\eta)(1-P_3)]^{n-2}\bigg[ \frac{\delta(\delta+\eta)}{\eta^2} -\frac{(\delta+\eta)^2}{\eta^2} \left(\frac{\delta}{\delta+\eta}\right)^{n-1}-\frac{\delta+\eta}{\eta}(n-2)\left(\frac{\delta}{\delta+\eta}\right)^{n-1}  \bigg] \notag \\
  {}&-(1-S)\delta(\delta+\eta)P_3^2[\delta(1-P_3)]^{n-2} \left[ \frac{\delta}{\eta}\left(\frac{\delta+\eta}{\delta}\right)^{n-1}-\frac{\delta+\eta}{\eta}-(n-2) \right]  \notag \\
  {}&+(1-S)\delta\eta P_3^2\frac{(n-1)(n-2)}{2}[\delta(1-P_3)]^{n-2}
\end{align}

Collecting all the results obtained in (16), (89) and (90), after careful calculations and numerical verification, we derive that
\begin{align}
&\Pr\{\Delta_B^P=n\}  \notag \\
={}& \frac{(\delta+\eta)P_3(pP_1+\eta)}{\delta-(\delta+\eta)(1-P_3)}\delta^{n-1} +\frac{P_3[\eta(\delta+\eta)+(1-S)(\delta+\eta)P_3(2\delta-\eta)]}{\eta(1-P_3)}[\delta(1-P_3)]^{n-1}  \notag \\
  {}& -\bigg\{\frac{pP_1\eta(1-P_3)+\delta\eta P_3}{[\delta-(\delta+\eta)(1-P_3)](1-P_3)} +\frac{(1-S)\delta P_3[\eta+(\delta+\eta)P_3]+\widetilde{S}(\delta+\eta)P_3}{\eta(1-P_3)} \bigg\} [(\delta+\eta)(1-P_3)]^{n-1}  \notag \\
  {}& +\frac{(1-S)(\delta+\eta)P_3^2}{1-P_3}n[\delta(1-P_3)]^{n-1} +\frac{\widetilde{S}P_3}{1-P_3}n[(\delta+\eta)(1-P_3)]^{n-1}  \qquad (n\geq3)
\end{align}

It can be checked directly that equation (89) is zero when $n=1$, and (90) equals zero for both cases $n=1$ and $n=2$. Therefore, the probability distribution (91) is in fact valid for all $n\geq1$. This completes the proof of Theorem 9.

\section{Proof of Corollary 1}
Notice that for $0<x<1$ and an arbitrary integer $k\geq1$, we have
\begin{equation}
\sum\nolimits_{n=1}^{\infty}n(n-1)\cdots(n-k+1)x^{n-k}=\frac{k!}{(1-x)^{k+1}}
\end{equation}

Then, according to equation (5), the mean of AoI $\Delta_{NB}^P$ is given as
\begin{align}
\mathbb{E}[\Delta_{NB}^P]&=\sum\nolimits_{n=1}^{\infty}n\Pr\{\Delta=n\}  \notag \\
&=\sum\nolimits_{n=1}^{\infty}\bigg\{\frac{(1-pP_1)P_3(1-\delta)}{(1-pP_1)P_3-p(1-P_1)P_2}n\delta^{n-1} + \xi n[(1-pP_1)(1-P_3)]^{n-1} \bigg\}   \notag \\
&=\frac{(1-pP_1)P_3}{[(1-pP_1)P_3-p(1-P_1)P_2](1-\delta)} + \frac{\xi}{[1-(1-pP_1)(1-P_3)]^2}
\end{align}

From AoI distribution (11), we can obtain that
\begin{align}
\mathbb{E}[\Delta_{NB}^{NP}]&=\left(\beta_1 + \frac{p(1-P_1)P_2}{\delta-(1-pP_1)(1-P_3)}\beta_1\right) \frac{1}{(1-\delta)^2}  \notag \\
& \quad  + \left(\beta_2 - \frac{p(1-P_1)P_2}{\delta-(1-pP_1)(1-P_3)}\beta_1\right) \frac{1}{[1-(1-pP_1)(1-P_3)]^2}+\frac{2p(1-P_1)P_2\beta_2}{[1-(1-pP_1)(1-P_3)]^3}   \notag \\
&=\frac{\beta_1}{(1-\delta)^2}+\frac{\beta_2}{[1-(1-pP_1)(1-P_3)]^2}  \notag \\
& \quad + \frac{p(1-P_1)P_2\beta_1[2-\delta-(1-pP_1)(1-P_3)]}{(1-\delta)^2 [1-(1-pP_1)(1-P_3)]^2} + \frac{2p(1-P_1)P_2\beta_2}{[1-(1-pP_1)(1-P_3)]^3}
\end{align}

Finally, the average value of AoI $\Delta_B^P$ can be calculated using distribution expression (21). We show that
\begin{align}
\mathbb{E}[\Delta_B^P]&=\frac{(\delta+\eta)P_3(pP_1+\eta)}{\delta-(\delta+\eta)(1-P_3)}\frac{1}{(1-\delta)^2}  -\frac{c_1}{[1-(\delta+\eta)(1-P_3)]^2} + \frac{c_2}{[1-\delta(1-P_3)]^2}  \notag \\
& \qquad + \frac{(1-S)(\delta+\eta)P_3^2}{1-P_3}\frac{1+\delta(1-P_3)}{[1-\delta(1-P_3)]^3} + \frac{P_3\widetilde{S}}{1-P_3}\frac{1+(\delta+\eta)(1-P_3)}{[1-(\delta+\eta)(1-P_3)]^3}     \\
&=\frac{(\delta+\eta)P_3(pP_1+\eta)}{[\delta-(\delta+\eta)(1-P_3)](1-\delta)^2}-\frac{c_1}{[1-(\delta+\eta)(1-P_3)]^2} +\frac{c_2}{[1-\delta(1-P_3)]^2}  \notag \\
& \qquad +\frac{(1-S)(\delta+\eta)P_3^2[1+\delta(1-P_3)]}{(1-P_3)[1-\delta(1-P_3)]^3} +\frac{P_3\widetilde{S}[1+(\delta+\eta)(1-P_3)]}{(1-P_3)[1-(\delta+\eta)(1-P_3)]^3}
\end{align}
where in (96) observing that
\begin{align}
\sum\nolimits_{n=1}^{\infty}n^2[\delta(1-P_3)]^{n-1} &= \sum\nolimits_{n=1}^{\infty}[n(n-1)+n][\delta(1-P_3)]^{n-1} \notag \\
&=\sum\nolimits_{n=1}^{\infty}n(n-1)[\delta(1-P_3)]^{n-1} + \sum\nolimits_{n=1}^{\infty}n[\delta(1-P_3)]^{n-1} \notag \\
&=\delta(1-P_3)\frac{2}{[1-\delta(1-P_3)]^3} + \frac{1}{[1-\delta(1-P_3)]^2}  \notag \\
&= \frac{1+\delta(1-P_3)}{[1-\delta(1-P_3)]^3}  \notag
\end{align}

Similarly, we have
\begin{equation}
\sum\nolimits_{n=1}^{\infty}n^2[(\delta+\eta)(1-P_3)]^{n-1} = \frac{1+(\delta+\eta)(1-P_3)}{[1-(\delta+\eta)(1-P_3)]^3}  \notag
\end{equation}

All the parameters can be found in previous Theorem 3, Theorem 6 and Theorem 9. This completes the proof of this Corollary.




%




\end{document}